\documentclass[onecolumn, 11 point, journal]{IEEEtran}

\usepackage{color}
\usepackage{times}
\usepackage{amssymb}
\usepackage{amsmath}
\usepackage{epsfig}
\usepackage{cite}
\usepackage{verbatim}
\usepackage{enumerate}
\usepackage{enumitem}

\usepackage{amsmath}
\usepackage{graphicx}
\usepackage{subfig}

\usepackage{epsf,bm}
\usepackage{epstopdf}
\usepackage{setspace}
\usepackage{algpseudocode}
\usepackage{algorithm}
\usepackage{cases}
\usepackage{tikz}
\usepackage{url}
\usepackage{blkarray}

\makeatletter 
\newbox\BA@first@box
\makeatother
\newtheorem{theorem}{\bf Theorem}[section]

\newtheorem{corollary}[theorem]{\bf Corollary}
\newtheorem{lemma}[theorem]{\bf Lemma}

\newtheorem{definition}{\bf Definition}[section]
\newtheorem{remark}{\bf Remark}[section]
\newtheorem{claim}[theorem]{\bf Claim}

\newcommand{\Head}{\text{Head}}
\newcommand{\Tail}{\text{Tail}}
\newcommand{\In}{\text{In}}
\newcommand{\Out}{\text{Out}}
\newcommand{\Ord}{\text{Ord}}

\newcommand{\qedw}{\hfill \ensuremath{\Box}}
\newcommand{\qed}{\nobreak \ifvmode \relax \else
      \ifdim\lastskip<1.5em \hskip-\lastskip
      \hskip1.5em plus0em minus0.5em \fi \nobreak
      \vrule height0.75em width0.5em depth0.25em\fi}
      
      \makeatletter
\newcommand*\bigcdot{\mathpalette\bigcdot@{.5}}
\newcommand*\bigcdot@[2]{\mathbin{\vcenter{\hbox{\scalebox{#2}{$\m@th#1\bullet$}}}}}
\makeatother
\newenvironment{proof}[1][Proof:]{\begin{trivlist}
\item[\hskip \labelsep {\bfseries #1}]}{\qedw\end{trivlist}}

\begin{document}
\fontfamily{cmr} \selectfont
\title{Linear Network Coding for Two-Unicast-$Z$ Networks: A Commutative Algebraic Perspective and Fundamental Limits}
\author{\large Mohammad Fahim and Viveck R. Cadambe \\ [.05in] 
\small  
\begin{tabular}{c} 
Department of Electrical Engineering,
Pennsylvania State University. \\
Email: fahim@psu.edu, viveck@engr.psu.edu.
\end{tabular} 
}
\maketitle
\vspace{-0pt}
\allowdisplaybreaks{
\begin{abstract}
 We consider a two-unicast-$Z$ network over a directed acyclic graph of unit capacitated edges; the two-unicast-$Z$ network is a special case of two-unicast networks where one of the destinations has apriori side information of the unwanted (interfering) message. In this paper, we settle open questions on the limits of network coding for two-unicast-$Z$ networks by showing that the generalized network sharing bound is not tight, vector linear codes outperform scalar linear codes, and non-linear codes outperform linear codes in general. We also develop a commutative algebraic approach to deriving linear network coding achievability results, and demonstrate our approach by providing an alternate proof to the previous results of C. Wang et. al., I. Wang et. al. and Shenvi et. al. regarding feasibility of rate $(1,1)$ in the network. 
\end{abstract}
\vspace{-0pt}
\section{Introduction}
\makeatletter{\renewcommand*{\@makefnmark}{}
\footnotetext{\hrule \vspace{0.05in} This work is supported by NSF grant No. CCF  1464336.

A short version of this work is published  in the Proceedings of  The IEEE International Symposium on Information Theory (ISIT), June 2017 \cite{fah-cad}.
}\makeatother }

There is significant interest in multiple unicast network coding and index coding in recent times. In addition to capturing the essence of network communication, there are interesting connections between special instances of the multiple unicast network communication problem and several emerging applications including topological interference management in wireless networks \cite{Jafar_topological}, codes for caching and content distribution \cite{maddahali_caching}, the index coding problem \cite{equiv-index}, and regenerating and locally recoverable codes for distributed storage \cite{Cadambe_asymptotic, Dimitris_LRC}. While the classical max-flow-min-cut theorem demonstrates the capacity of the single unicast problem \cite{shannon1956}, \cite{ford}, even the two-unicast problem is a notoriously challenging open problem in network information theory \cite{Kamath-hard}, \cite{MissImp}.  

In this paper, we study the most simple multiple unicast communication scenario, in terms of message structure, whose capacity is unknown: the two-unicast-$Z$ network. 
The two-unicast-$Z$ network, like the two-unicast network, has two independent message sources and two destinations, each destination respectively requiring to decode one of the two message sources. One of the two destinations, say the second destination, has apriori side information of the unintended (first) message source (See Fig. \ref{fig:networksketch}). Like the $Z$-interference channel in wireless communications, the two sources of the network interfere at only one destination. The study of two-unicast-$Z$ networks is important, because, like index coding and other simplified variants, insights obtained through code development for two-unicast-$Z$ networks can potentially influence code design for more general multiple unicast networks and its related applications.

Unlike the two-unicast network \cite{Kamath_twounicast},\cite{shenvi2009}, \cite{wang-shroff},\cite{CCWang_twounicast},\cite{rate11-tight}, where (a) linear network coding is insufficient for capacity, (b) vector linear codes outperform scalar linear codes, and (c) the generalized network sharing (GNS) cut set bound is not  tight in general, the question of whether non-linear network coding, vector linear codes, or bounds stronger than the GNS bound are required to characterize the achievable rate region for \emph{two-unicast-$Z$} networks was open.  In particular, because two-unicast-$Z$ networks are a special case of two-unicast networks, the results which demonstrate the insufficiency of scalar linear and non-linear codes and the GNS bounds, for two-unicast networks, do not naturally extend to two-unicast-$Z$ networks. In fact, a previous work \cite{zeng2013z} developed a special class of two-unicast-$Z$ networks where the generalized network sharing bound is shown to be tight.  
 
 In this paper, we resolve these open questions for two-unicast-$Z$ networks. In particular, we show that for two-unicast-$Z$ networks, (a) vector linear codes outperform scalar linear codes, (b) non-linear codes outperform linear codes, and (c) that the GNS bound is not tight in general. Our impossibility results come from construction of specific network instances where these open questions are resolved through gaps between the specified achievable schemes, in the case of scalar and vector linear codes, or the converse, in the case of the GNS bound, and an optimal achievable rate.

A second contribution of this paper is the development of a commutative algebraic perspective of linear network coding. An algebraic framework for network coding has been established in \cite{Koetter_Medard} where scalar linear solvability over a general network is cast as a polynomial solvability problem. An interesting converse result in \cite{Dough-eq}, \cite{Dough-eq-t} has shown that for any collections of polynomials there exists a solvable equivalent directed acyclic network. This result implies that  the complexity of determining whether networks are scalar-linearly solvable over particular finite fields and the complexity of determining whether collections of polynomial are solvable over the corresponding fields are the same. An algebraic formulation based on path gains has been introduced in \cite{alg-form}, \cite{new-per}, \cite{path-gain}; these references also present an algorithm that casts  scalar linear coding solvability of a network as solving a set of polynomial equations with only linear and quadratic terms. In the context of two-unicast-$Z$ networks, a low-complexity heuristic for linear network coding has been developed in \cite{zeng2016alignment}.

Our starting point is the algebraic framework of network coding \cite{Koetter_Medard}. We develop fundamental connections between the polynomials formed by the local coding co-efficients and properties of the network communication graph. 
We describe our perspective through an alternate proof, for two-unicast-$Z$ networks, of the result of \cite{wang-shroff},   \cite{CCWang_twounicast}, \cite{rate11-tight}, and \cite{shenvi2009}, which establishes the feasibility of rate $(1,1)$ for two-unicast networks. In particular, \cite{wang-shroff},   \cite{CCWang_twounicast}, \cite{rate11-tight},  \cite{shenvi2009} give path-based necessary and sufficient conditions on the achievability of rate $(1,1)$ in two-unicast networks. An implication of these results is that the rate tuple $(1,1)$ is achievable if and only if the generalized network sharing cut set bound  \cite{Kamath_twounicast} is at least $2$, and the individual source destination pairs have a cut of at least $1$. Our alternate proof, albeit for the special case of two-unicast-$Z$ networks, encompasses new ideas and methods. 

Our approach is to write the solvability criterion based on the Nullstellensatz as per \cite{Koetter_Medard}, and then infer the final result based on elementary properties on the degrees of the polynomials that participate in the solvability criterion.  Among others, one interesting by-product of our analysis is the discovery of a network decomposition lemma. In linear network coding, the effect of a path to the overall transfer function is the product of the local coding weights at each edge in that path. Thus, given any edge in the network, the gain of all source-destination paths that flow through that edge can be factorized as the product of the gain from the source to that edge, and the gain from the edge to the destination. If the edge happens to be a source-destination cut, then this factorization is, in fact, a factorization of source-destination transfer matrix. The network decomposition lemma is a generalization of factorization for the case where the cut can involve multiple edges. Given a collection of edges that forms a cut, we effectively factorize the source-destination transfer matrix as a product of two transfer matrices: one from the source to the collection of edges, and another from the collection of edges to the destination. A non-trivial  technical hurdle that is absent in the single edge case, but we solve for cuts consisting of possibly multiple edges, is to carefully decouple the effect of the paths \emph{among} the cut-edges in the final factorization. 

The paper is organized as follows, we describe the system model in Section \ref{sec:background}. Afterwards, a brief background on commutative algebra is introduced in Section \ref{sec:commbkgnd} followed by a necessary and sufficient condition for the achievability of rate $(1,1)$ described in Section \ref{sec:cons} as a consequence to Hilbert's Nullstellensatz. We develop a network decomposition lemma in Section \ref{sec:netdecomp} and combine this lemma with tools from commutative algebra to derive the achievability proof. We present our achievability proof in Section \ref{sec:intermediate}. Proofs  of our impossibility results describing the insufficiency  of the GNS bound and linear network coding in two-unicast-$Z$ networks are provided in Sections \ref{sec:GNS} and \ref{sec:nonlinear}, respectively. We conclude with a discussion on challenges and open problem related to expanding our commutative algebraic approach to networks beyond the two-unicast-$Z$  network in Section \ref{sec:conc}.

\section{System Model}\label{sec:background}
Throughout this paper,  $\mathbb{Z}_{\geq 0}$ denotes the set of non-negative integers, and $\mathbb{Z}_{+}$ denotes the set of positive integers.
We consider a directed acyclic graph (DAG) $\mathcal{G} = (\mathcal{V},
\mathcal{E})$, where $\mathcal{V}$  denotes the set of vertices and
$\mathcal{E}$ denotes the set of edges. We assume unit-capacity edges and allow multiple edges between
vertices, hence, $\mathcal{E} \subset \mathcal{V} \times \mathcal{V} \times
\mathbb{Z}_{+}$.
For an edge $e = (u, v, i) \in \mathcal{E}$, we denote $\text{Head}(e)=v$ and
$\text{Tail}(e) = u$. For a given vertex $v \in \mathcal{V}$, we denote $\In(v)
= \left\{ e\in \mathcal{E} : \Head(e) = v\right\} $ and $\Out(v) =
\left\{e\in\mathcal{E}: \Tail(e) = v \right\}$.   Moreover, for an edge $e$,  we denote $\In(e)
= \left\{ e'\in \mathcal{E} : \Head(e') = \Tail(e)\right\} $ and $\Out(e) =
\left\{e'\in\mathcal{E}: \Tail(e') = \Head(e) \right\}$.

A path $p$ is a sequence of edges $(e_{m_1},e_{m_2},\ldots, e_{m_l})$ where  $\text{Head}(e_{m_i}) = \text{Tail}(e_{m_{i+1}})$ for $i=1,2,\ldots,l-1$. Let $\mathcal{E}_{1}, \mathcal{E}_{2},\mathcal{E}_{3} \subseteq \mathcal{E}$,  $\mathcal{E}_{1} \rightarrow \mathcal{E}_{2}$ denotes the set of all paths from $e_{1}$ to $e_{2}$ such that $e_{i}\in \mathcal{E}_{i},i=1,2$. Also  $\mathcal{E}_{1} \rightarrow \mathcal{E}_{2} \backslash \mathcal{E}_{3}$ denotes the set of all paths from $e_1$ to $e_2$   that do not contain any of the edges in $\mathcal{E}_{3}$ where $e_{i}\in \mathcal{E}_{i},i=1,2$. If $\mathcal{E}_{1}=\{e_1\},\mathcal{E}_{2}=\{e_2\}, \mathcal{E}_{3}=\{e_3\}$ are singletons, then we simply write $e_{1} \rightarrow e_2$ or $e_{1} \rightarrow e_{2} \backslash e_{3}$ as the case may be.
Because $\mathcal{G}$ is a DAG, there is a topological ordering $\Ord:\mathcal{E}\rightarrow \mathbb{Z}_{+}$ on the edges of the graph with the property that $e_{1} \rightarrow e_{2} \Rightarrow \Ord(e_{1}) < \Ord(e_2)$.

\subsection{Algebraic framework for linear network coding}\label{sec:algfrmwrk}

\begin{figure}[ht]
    \centering
    \includegraphics[scale=0.35]{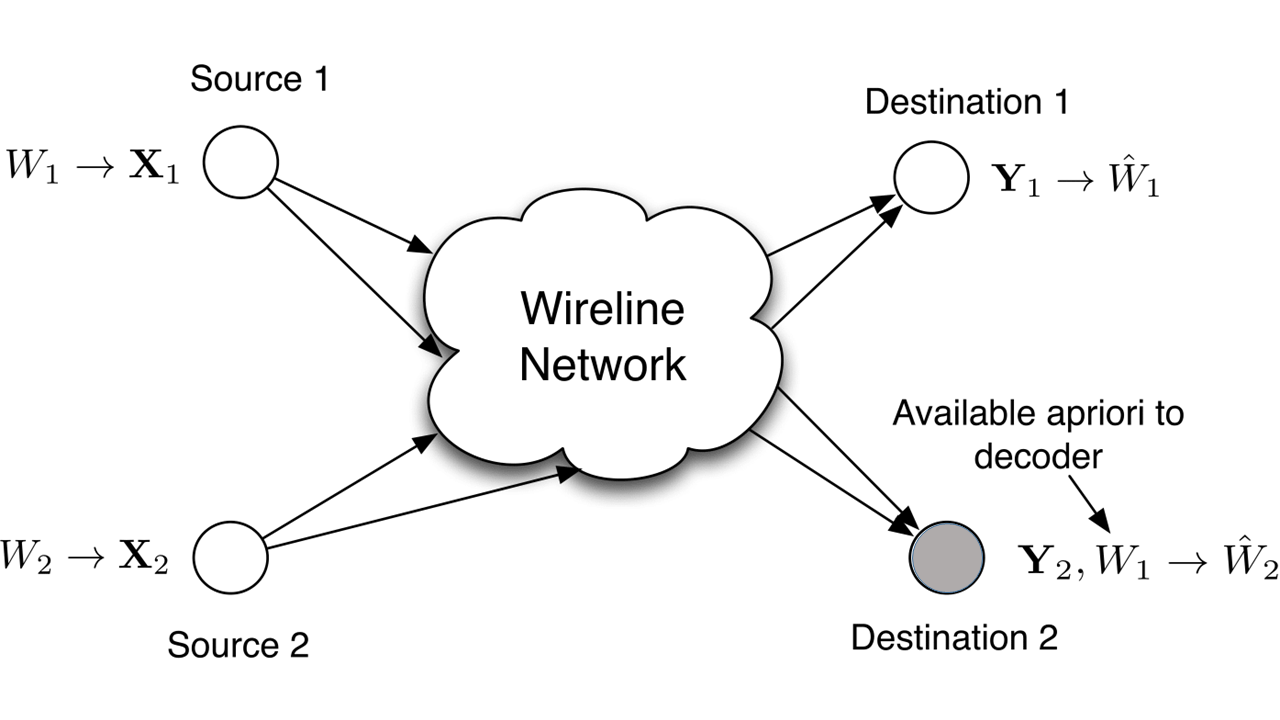}
    \caption{A two-unicast-$Z$ network}
    \label{fig:networksketch}
%\vspace{-15pt}
\end{figure}
We set up linear network coding schemes based on the algebraic framework of \cite{Koetter_Medard}. Let $\mathbb{K}$ denote the algebraic closure of the field $\mathbb{F}_2.$ Let $e_1, e_2, \dots, e_{|\mathcal{E}|}$ denote the edges of $\mathcal{E}$ in topological order, i.e., $i < j \Leftrightarrow \Ord(e_{i}) < \Ord(e_{j})$ . The \emph{local coding matrix} is an upper triangular matrix $\mathbf{F}^{\mathcal{G}}$
 whose
element in the $i$-th row and $j$-th column $F_{i,j}^{\mathcal{G}}$ is given as,
\begin{align}
    F^{\mathcal{G}}_{i,j} = \left\{
  \begin{array}{cl}
    \beta_{e_i, e_j} & \text{if Head}(e_i) = \text{Tail}(e_j) \\
    0 & \text{otherwise},
  \end{array}
\right.
\end{align}
where $\beta_{e_i,e_j}$ is a variable that represents the local coding coefficient relating  $e_i$ to $e_j$. Wherever the graph $\mathcal{G}$ being considered is clear, we will simply omit the superscript and simply express the local coding matrix as $\mathbf{F}.$  We denote the set whose elements are the (non-zero) entries of $\mathbf{F}$ as $\bar{\mathbf{F}}$. That is $\bar{\mathbf{F}}=\{\beta_{e_i,e_j}: \text{Head}({e_i})=\text{Tail}({e_j}) \text{ and } e_i,e_j\in \mathcal{E}\}$. Note that $\bar{\mathbf{F}}$ is the set of all local coding coefficients. We denote the polynomial ring with field $\mathbb{K}$ and set of variables  $\bar{\mathbf{F}}$ as $\mathbb{K}[\mathbf{F}]$. 

For a path $p = (e_{m_1},e_{m_2},\ldots, e_{m_l})$, the \emph{weight} of the path is a function that maps the path to an element of $\mathbb{K}[\mathbf{F}]$ defined as $w(p) = \prod_{i=1}^{l-1} \beta_{e_{m_i},e_{m_{i+1}}}.$  For two edges $e_{i},e_{j},$ let $H_{i,j}(\mathbf{F}) = \sum_{p \in e_i \rightarrow e_j} w(p).$ The network extended transfer matrix $\mathbf{H}(\mathbf{F})$ is a $|\mathcal{E}| \times |\mathcal{E}|$ matrix whose entry in the $i$-th row and the $j$-th column is $H_{i,j}$. Note that every element of $\mathbf{H}(\mathbf{F})$ lies in the polynomial ring $\mathbb{K}[\mathbf{F}]$. It can be shown that $\mathbf{H}(\mathbf{F}) = (\mathbf{I}-\mathbf{F})^{-1},$ where $\mathbf{I}$ is the $|\mathcal{E}|\times |\mathcal{E}|$ identity matrix in $\mathbb{K}$ \cite{Koetter_Medard}.

A scalar linear network code is specified by a  matrix $\mathbf{F}^{*}=\mathcal{C}(\mathbf{F})$, where $\mathcal{C}$ is a mapping $\mathcal{C}:\mathbb{K}[\mathbf{F}]^{|\mathcal{E}|\times|\mathcal{E}|} \rightarrow \mathbb{K}^{|\mathcal{E}|\times|\mathcal{E}|}$. The network extended transfer matrix, for this specific scalar linear network code, is simply obtained by evaluating the corresponding polynomials, $\mathbf{H}(\mathbf{F}^{*})$.

\subsubsection*{\textbf{Algebraic framework for vector linear network coding}}
For ease of exposition, the algebraic framework described above is specified for scalar linear network coding schemes. However, the framework can be extended to vector linear network coding schemes as well. For a vector linear network code with vectors of dimension $v > 1$, the local coding matrix is a $v|\mathcal{E}|\times v|\mathcal{E}|$ upper triangular matrix $\mathbf{F}=(\mathbf{F}_{i,j})$  whose $(i,j)$-th  block (submatrix) of dimension $v \times v$ is
\begin{align}
    \mathbf{F}_{i,j} = \left\{
  \begin{array}{cl}
   \mathbf{B}_{e_i, e_j} & \text{if Head}(e_i) = \text{Tail}(e_j) \\
    0 & \text{otherwise},
  \end{array}
\right.
\end{align}
where $\mathbf{B}_{e_i,e_j}$ is a matrix of variables $\beta_{e_i,e_j}^{k,l}$ with $k,l \in \{1, \cdots, v\}$ with the variable $\beta_{e_i,e_j}^{k,l}$ in the $k$-th row and $l$-th column of $\mathbf{B}_{e_i,e_j}$  and the variables $\beta_{e_i,e_j}^{k,l}$ for $k,l \in \{1, \cdots, v\}$  represent the local coding coefficients relating  $e_i$ to $e_j$. The notions of the weight of a path and the network extended transfer matrix change accordingly. That is, for a path $p = (e_{m_1},e_{m_2},\ldots, e_{m_l})$, the \emph{weight} of the path is a function that maps the path to an element of $\mathbb{K}[\mathbf{F}]^{v \times v}$ defined as $w(p) = \prod_{i=1}^{l-1} \mathbf{B}_{e_{m_i},e_{m_{i+1}}}.$  For two edges $e_{i},e_{j},$ let $\mathbf{H}_{i,j}(\mathbf{F}) = \sum_{p \in e_i \rightarrow e_j} w(p).$ The network extended transfer matrix $\mathbf{H}(\mathbf{F})$ is a $v|\mathcal{E}| \times v|\mathcal{E}|$ matrix  whose $(i,j)$-th  block (submatrix) of dimension $v \times v$ is $\mathbf{H}_{i,j}$.

A vector linear network code is specified by a  matrix $\mathbf{F}^{*}=\mathcal{C}_{vec}(\mathbf{F})$, where $\mathcal{C}_{vec}$ is a mapping $\mathcal{C}_{vec}:\mathbb{K}[\mathbf{F}]^{v|\mathcal{E}|\times v|\mathcal{E}|} \rightarrow \mathbb{K}^{v|\mathcal{E}|\times v|\mathcal{E}|}$. The network extended transfer matrix, for this specific vector linear network code, is simply obtained by evaluating the corresponding polynomials, $\mathbf{H}(\mathbf{F}^{*})$. 
\subsection{Two-unicast-$Z$ network}

We depict a two-unicast-$Z$ network in Fig. \ref{fig:networksketch}. Note that, throughout this paper, we shade the color of the destination node that possesses the unintended message source as side information. 
\begin{definition}[Two-Unicast-$Z$ Network]
A $(\mathcal{G},\mathcal{S}_{1},\mathcal{T}_{1},\mathcal{S}_{2},\mathcal{T}_{2})$ two unicast-$Z$ network consists of a graph $\mathcal{G}=(\mathcal{V},\mathcal{E})$, and sets $\mathcal{S}_{1},\mathcal{S}_{2},\mathcal{T}_{1},\mathcal{T}_{2} \subseteq \mathcal{E}$. For $i \in \{1,2\}$ the sets $\mathcal{S}_{i},\mathcal{T}_{i}$ are respectively referred to the edges of the $i$-th source and destination, respectively. 
\end{definition}

Throughout this paper, for $i \in \{1,2\}$, we denote the node $v \in \mathcal{V}$ such that $\In(v)=\mathcal{S}_i$ by \emph{Source} $i$. Similarly, for $i \in \{1,2\}$, we denote the node $v \in \mathcal{V}$ such that $\Out(v)=\mathcal{T}_i$ by \emph{Destination} $i$. Whenever it is clear, throughout this paper, we omit the edges of $\mathcal{S}_i, \mathcal{T}_i, i\in\{1,2\}$ in the figures of the two-unicast-$Z$ networks and just keep the nodes \emph{Source} $i$, \emph{Destination} $i$, $i \in \{1,2\}$.

\subsection{Achievability of rate $(R_1,R_2)$ in the two-unicast-$Z$ network }
For a $(\mathcal{G},\mathcal{S}_{1},\mathcal{T}_{1},\mathcal{S}_{2},\mathcal{T}_{2})$ two-unicast-$Z$ network, a rate  $(R_1, R_2)$  is achievable if, for any $\epsilon \geq 0$, there exist a positive integer block length $n$, a finite alphabet $\mathcal{A}$, local encoding functions: 
\begin{itemize}
\item $f_{e_i}: \mathcal{A}^{\lceil n R_i \rceil} \rightarrow \mathcal{A}^n$, for any $e_i \in \Out(\mathcal{S}_i), i\in \{1,2\}$, and

\item $f_{e}: \mathcal{A}^{n |\In(e)|} \rightarrow \mathcal{A}^n$, for any $e \in \mathcal{E} - (\Out(\mathcal{S}_1) \cup \Out(\mathcal{S}_2) \cup \mathcal{T}_1 \cup \mathcal{T}_2)$,
\end{itemize}
 and local decoding functions  $g_{t_i}: \mathcal{A}^{n |\In(t_i)|} \rightarrow \mathcal{A}^n$, for any $t_i \in \mathcal{T}_i, i\in \{1,2\}$ such that, under the uniform probability distribution of $\mathcal{A}^{\lceil nR_1\rceil} \times \mathcal{A}^{\lceil nR_2\rceil}$,   $\text{Pr}(g(w_1,w_2)) \neq (w_1,w_2)) \leq \epsilon$, where $g: \mathcal{A}^{\lceil nR_1\rceil} \times \mathcal{A}^{\lceil nR_2\rceil} \rightarrow \mathcal{A}^{\lceil nR_1\rceil} \times \mathcal{A}^{\lceil nR_2\rceil}$ is a global decoding function induced by the local encoding and decoding functions.  Moreover, the closure of the set of all achievable rates in a network is called \emph{the rate region} of the network, and the supremum of the rate region is called \emph{the capacity} of the network.  
 
 If rate $(R_1,R_2)$ is achievable in a two-unicast-$Z$ network with $\epsilon =0$, then, we say that rate $(R_1,R_2)$ is \emph{zero-error achievable} in this network. Similarly, the closure of the set of all zero-error achievable rates in a network is called \emph{the zero-error rate region} of the network, and the supremum of the zero-error rate region is called \emph{the zero-error capacity} of the network.

 Notice that, for any coding scheme that achieves rate $(R_1,R_2)$ in a two-unicast-$Z$ network, if the encoding and decoding functions are linear, then the coding scheme is \emph{linear} and rate $(R_1,R_2)$ is said to be \emph{linearly achievable} in the two-unicast-$Z$ network.
 Specifically,  we describe next the linear achievability of rate $(R_1,R_2)$ in the two-unicast-$Z$ network under the algebraic framework described in Section \ref{sec:algfrmwrk}. 

\subsubsection*{\textbf{Linear achievability of rate ${(R_1,R_2)}$ in the two-unicast-$Z$ network}}
Consider the linear network coding algebraic framework for a two-unicast-$Z$ network with vectors of dimension $v \geq 1$. Note that $v=1$ corresponds to the scalar linear network coding framework. Let the two-unicast-$Z$ network has sources $\mathcal{S}_{i},i=1,2$ and destinations $\mathcal{T}_{i},i=1,2$.  For the source edge set $\mathcal{S}_{i}, i=1,2$ and the destination edge set $\mathcal{T}_{j},j=1,2,$ the transfer matrix $\mathbf{G}_{i,j}(\mathbf{F})$ is a $v|\mathcal{S}_{i}| \times v|\mathcal{T}_{j}|$ matrix with entries in $\mathbb{K}[\mathbf{F}]$ whose rows (columns) are the  rows (columns) of $\mathbf{H}$  corresponding to $\mathcal{S}_{i}$ ($\mathcal{T}_{j}$).

We now define the notion of achievability  of rate $(R_1,R_2)$ via  linear coding  in the two-unicast-$Z$ network. We assume that source $\mathcal{S}_1$ wants to convey a message $W_1$ of $R_1$ symbols to destination $\mathcal{T}_1$. Similarly, we assume that source $\mathcal{S}_2$ wants to convey a message $W_2$ of $R_2$ symbols to destination $\mathcal{T}_2$.  To understand our definition of achievability of $(R_1,R_2)$ via  linear coding, it is useful to imagine source vectors $\mathbf{X}_{1}=(\mathbf{X}_{1,1}, \cdots, \mathbf{X}_{1,R_1}),\mathbf{X}_{2}=(\mathbf{X}_{2,1}, \cdots, \mathbf{X}_{2, R_2})$ with entries in $\mathbb{K}$ of dimensions $1 \times v R_1$ and $1 \times v R_2$, respectively, where $\mathbf{X}_{i,j_i}, j_i\in\{1,\cdots, R_i\},i \in \{1,2\}$ are vectors of dimension $1 \times v$ such that  $\mathbf{X}_{1,j}$, $j\in\{1,\cdots, R_1\}$, represents the $j$-th symbol sent by the first source, and   $\mathbf{X}_{2,j}$, $j\in\{1,\cdots, R_2\}$, represents the $j$-th symbol sent by the second source. The goal of a network coding scheme is to convey these vectors to their respective destinations.

It is worth noting that for a  linear  coding scheme, there is no loss in generality in assuming that $|\mathcal{S}_{1}|=|\mathcal{T}_{1}|=R_1$ and $|\mathcal{S}_{2}|=|\mathcal{T}_{2}|=R_2$. Clearly, for a  linear coding problem, $|\mathcal{S}_i|$ cannot exceed $|\mathcal{T}_i|$, for $i=1,2$, otherwise the destination with less cardinality than its associated source receives an underdetermined set of equations on its destination edges. On the other hand, if $|\mathcal{S}_{i}|<|\mathcal{T}_{i}|$ for some $i\in\{1,2\}$, then the achievability of rate $|\mathcal{S}_{i}|$ on the communication session between source $\mathcal{S}_i$ and destination $\mathcal{T}_i$ requires the achievability of rate $|\mathcal{S}_{i}|$ on the communication session between source $\mathcal{S}_i$ and  a subset of destination edges  $\tilde{\mathcal{T}}_i$ for some $\tilde{\mathcal{T}}_i \subset {\mathcal{T}}_i$ with $|\tilde{\mathcal{T}}_i| = |{\mathcal{S}}_i|$.

For a specific linear network coding scheme with $\mathbf{F}^{*} \in \mathbb{K}^{v|\mathcal{E}|\times v|\mathcal{E}|}$, the vectors $\mathbf{Y}_{1},\mathbf{Y}_{2}$ received respectively by the two receivers in a two-unicast-$Z$ network can be written as 
\begin{eqnarray} \mathbf{Y}_1 &=& \mathbf{X}_{1}\mathbf{G}_{1,1}(\mathbf{F}^{*})+\mathbf{X}_{2}\mathbf{G}_{2,1}(\mathbf{F}^{*})\\
\mathbf{Y}_2 &=& \mathbf{X}_{2}\mathbf{G}_{2,2}(\mathbf{F}^{*}).
\end{eqnarray}
We have not written the effect of $\mathbf{X}_{1}$ at receiver $2$, since the receiver can subtract the effect of $\mathbf{X}_{1}$ from the side information that it possesses.
 We refer to the linear coding scheme $\mathbf{F}^{*}$ as an achievable scheme if $\mathbf{X}_{1},\mathbf{X}_{2}$ are recoverable from $\mathbf{Y}_{1},\mathbf{Y}_{2}$, respectively. For successful recovery of the two sources from the respective destinations, we require
\begin{eqnarray}
&\det\left(\mathbf{G}_{i,i}(\mathbf{F}^{*})\right) \neq 0, i =1,2, ~~~~\mathbf{G}_{2,1}(\mathbf{F}^{*}) = \mathbf{0}_{v R_2 \times v R_1}&
\label{eq:achievability}
\end{eqnarray}
Note that the above conditions are necessary and sufficient since we restricted $|\mathcal{S}_{1}|=|\mathcal{T}_{1}|=R_1, |\mathcal{S}_{2}|=|\mathcal{T}_{2}|=R_2$. We now define our notion of linear achievability formally.

\begin{definition}[Linear achievability of rate $(R_1,R_2)$]\label{def:scalarachievable}
In a $(\mathcal{G},\mathcal{S}_{1},\mathcal{T}_{1},\mathcal{S}_{2},\mathcal{T}_{2})$ two-unicast-$Z$ network with  $|\mathcal{S}_{1}|=|\mathcal{T}_{1}|=R_1, |\mathcal{S}_{2}|=|\mathcal{T}_{2}|=R_2,$ the rate $(R_1,R_2)$ is said to be achievable via  linear network coding with vectors of dimension $v \geq 1$, if there exists a  linear network coding scheme $\mathbf{F}^{*} \in \mathbb{K}^{v|\mathcal{E}|\times v|\mathcal{E}|}$ such that (\ref{eq:achievability}) holds. 
\end{definition}
\begin{remark}
Definition \ref{def:scalarachievable} applies to both scalar and vector linear achievability of rate $(R_1,R_2)$ in two-unicast-$Z$ networks. In particular, when $v=1$, the definition describes scalar linear achievability, and when $v>1$, the definition describes vector linear achievability.
\end{remark}

%In Section { \ref{sec:netdecomp}}, we describe the notion of achievability in the language of commutative algebra. 

An upper bound to the two-unicast-$Z$ problem is  the generalized network sharing (GNS) bound introduced in \cite{GNSbound}.

\subsection{GNS Bound}
Before introducing the GNS upper bound that has been proposed in \cite{GNSbound}, first, a generalized network sharing cut set  can be defined as follows.
\begin{definition}[The generalized network sharing (GNS) cut set\cite{GNSbound}] 
Let $(\mathcal{G},\mathcal{S}_{1},\mathcal{T}_{1},\mathcal{S}_{2},\mathcal{T}_{2})$ be a two-unicast-$Z$ network, a set $\mathcal{S}\subseteq \mathcal{E}$ is defined as a GNS cut set if $\mathcal{G} \backslash \mathcal{S}$ has no $\mathcal{S}_1 \rightarrow \mathcal{T}_1$ paths, no $\mathcal{S}_2 \rightarrow \mathcal{T}_2$ paths and no $\mathcal{S}_2 \rightarrow \mathcal{T}_1$ paths.
\end{definition}
It has been shown in \cite{GNSbound} that the minimum size of a GNS cut set in a two-unicast-$Z$ network provides an upper bound on the achievable sum-rates in this network. More formally, let $(\mathcal{G},\mathcal{S}_{1},\mathcal{T}_{1},\mathcal{S}_{2},\mathcal{T}_{2})$ be a two-unicast-$Z$ network with unit-capacity edges, $R_1+R_2 \leq |\mathcal{S}|$, where $(R_1,R_2)$ is any achievable rate   and $\mathcal{S}$ is any GNS cut set in the network. In addition, it has been shown in \cite{Kamath_twounicast}, \cite{GNS-np} that computing the GNS bound in multiple unicast networks, including two-unicast networks, is NP-hard.  

\begin{remark}[Notation]
	In the remainder of this paper, %we consider a two-unicast-$Z$ network with $\mathcal{S}_{1}=\{s_{1}\},\mathcal{T}_{1}=\{t_{1}\}, \mathcal{S}_{2}=\{s_2\}$ and $\mathcal{T}_{2}=\{t_2\}.$ We will assume that the min-cut between $\mathcal{S}_{1}$ and $\mathcal{T}_{1}$ is at least $1$ and the min-cut between $\mathcal{S}_{2}$ and $\mathcal{T}_{2}$ is at least $1$. Because of the max-flow min-cut theorem, this is equivalent to stating that $\mathbf{G}_{1,1}(\mathbf{F})\mathbf{G}_{2,2}(\mathbf{F}) \neq 0$.
	We drop the dependence on $\mathbf{F}$ with the understanding that, unless otherwise specified, all network transfer polynomials lie in the ring $\mathbb{K}(\mathbf{F}).$ In instances where we refer to a specific network code, $\mathbf{F}^{*} \in \mathbb{K}^{|\mathcal{E}|\times |\mathcal{E}|},$ we specify this explicitly.
\end{remark}

\begin{remark}
It is worth noting that we have chosen the field of operation $\mathbb{K}$ as the algebraic closure of $\mathbb{F}_{2}$ in the above definitions. The algebraic closure of $\mathbb{F}_{2}$ consists of every finite extension of $\mathbb{F}_{2}$ as a sub-field. It is therefore useful to note that as per Definition \ref{def:scalarachievable}, a rate $(R_1,R_2)$ is achievable via  linear coding if and only if there is \emph{some} finite extension of $\mathbb{F}_{2}$ over which the rate is achievable. It is also worth noting that there is, for general networks beyond two-unicast-$Z$ networks, a loss of generality in restricting to extensions of $\mathbb{F}_{2}$, since there exist networks where the notion of solvability depends on the characteristic of the field \cite{Dougherty_insufficiency}. However, the field characteristic does not influence the results of this paper, so we restrict ourselves to extensions of $\mathbb{F}_{2}$ in this document. 
\end{remark}
\begin{remark}[Notation]
Let $\mathcal{G}=(\mathcal{V},\mathcal{E})$ be a directed acyclic graph, $s,t,s_i,t_i \in \mathcal{E}$, $i \in\{1,2\}$, and let $\mathcal{U}_1, \mathcal{U}_2 \subseteq\mathcal{E}$.

\begin{itemize}
\item $\sum\limits_{p :  s \rightarrow t} \hspace{-1mm}w(p)$ denotes the sum of the weights of all  $s\rightarrow t$ paths and is called the transfer polynomial from $s$ to $t$.
\item $\sum\limits_{p: s \rightarrow t \text{ via } \mathcal{U}_1 \backslash \mathcal{U}_2} \hspace{-1mm}w(p)$ denotes the sum of the weights of $s\rightarrow t$ paths such that each of these $s \rightarrow t$ paths goes through at least on edge in $\mathcal{U}_1$ and does not go through any edge in $\mathcal{U}_2$.
\item  $\sum\limits_{\substack{p_1: s_1 \rightarrow t_1 \\  p_2:s_2 \rightarrow t_2}} \hspace{-3mm}w(p_1)w(p_2)$ denotes the sum of the weights of all $(p_1,p_2)$  pairs of paths such  that $p_1$ is an $s_1 \rightarrow t_1$ path and $p_2$ is an $s_2 \rightarrow t_2$ path.
\end{itemize}
\end{remark}
\section{Commutative Algebra Background}\label{sec:commbkgnd}
In this section, we describe some elementary concepts of commutative algebra \cite{Cox_Little}, and state a central  result: Hilbert's Nullstellensatz. Afterwards, in the next section, we state and describe conditions equivalent to (\ref{eq:achievability}) for achievability of rate $(1,1)$ as a corollary to Hilbert's Nullstellensatz. We begin with some definitions.

\begin{definition}[Ideals]
Let $\mathbb{K}$ be a field. A subset  $ I $ of the polynomial ring $\mathbb{K}[x_1, x_2, \cdots, x_n]$ is an ideal if it satisfies:
\begin{enumerate}[label=(\alph*)]
\item $0 \in I$,
\item if $f$, $g \in I $, then $f + g \in I$, and
\item if $f \in I$ and $h \in \mathbb{K}[x_1, x_2, \cdots, x_n]$, then $hf \in I$.  
\end{enumerate}
\end{definition}
\begin{definition}[Ideals generated by polynomials]
Let $\mathbb{K}$ be a field, and let $f_1, f_2, \cdots, f_m$ be polynomials in the polynomial ring $\mathbb{K}[x_1, x_2, \cdots, x_n]$. The ideal generated by polynomials   $f_1, f_2, \cdots, f_m$ in $\mathbb{K}[x_1, x_2, \cdots, x_n]$ is denoted as  $<f_1, f_2, \cdots, f_m>$ and defined as  $$<f_1, f_2, \cdots, f_m> = \left\{\sum\limits_{i=1}^m h_i f_i: h_1, h_2, \cdots, h_m \in \mathbb{K}[x_1, x_2, \cdots, x_n]\right\}.$$
\end{definition}
%\begin{remark}
%If $f_1, f_2, \cdots, f_m \in k[x_1, x_2, \cdots, x_n]$, then  $<f_1, f_2, \cdots, f_m>$ is an ideal of $ k[x_1, x_2, \cdots, x_n]$, and we say  $<f_1, f_2, \cdots, f_m>$ is the ideal generated by $f_1, f_2, \cdots, f_m $.
%\end{remark}
\begin{definition}[Affine varieties]
Let $\mathbb{K}$ be a field, and let $f_1, f_2, \cdots, f_m$ be polynomials in the polynomial ring $\mathbb{K}[x_1, x_2, \cdots, x_n]$. The affine variety denoted by $\mathbf{V}(f_1, f_2, \cdots, f_m) \subseteq \mathbb{K}^n$ is defined to be its set of ``roots'', that is,  $$\mathbf{V}(f_1, f_2, \cdots, f_m) = \{(a_1, a_2, \cdots, a_n) \in \mathbb{K}^n : f_i(a_1, a_2, \cdots, a_n) = 0 \hspace{2mm}\forall i \in \{1, \cdots, m\}\}.$$ 
\end{definition}
\begin{definition}[Ideals of varieties]
Let $\mathbb{K}$ be a field, $\mathbb{K}[x_1, x_2, \cdots, x_n]$ be its associated polynomial ring, and let $V \subset \mathbb{K}^n$ be an affine variety. The ideal of the variety $V$ is denoted as $\mathbf{I}(V)$ and defined as $$\mathbf{I}(V) = \left\{ f \in \mathbb{K}[x_1, x_2, \cdots, x_n]: f(a_1, a_2, \cdots, a_n ) = 0 \hspace{2mm}\forall (a_1, a_2, \cdots, a_n) \in V\right\}.$$
\end{definition}
%\begin{remark}
%If $V \subset k^n$ is an affine variety, then $\mathbf{I}(V) \subset k[x_1, x_2, \cdots, x_n]$  is an ideal, and we say $\mathbf{I}(V)$ is the ideal of $V$. 
%\end{remark}
\begin{remark}[A reversing-inclusion property\cite{Cox_Little}]\label{rev}
Let $\mathbb{K}$ be a field. Let $V$ and $W$ be affine varieties in $\mathbb{K}^n$. Then, $V \subseteq W$ if, and only if, $\mathbf{I}(V) \supseteq \mathbf{I}(W)$.
\end{remark}
\begin{theorem}[Hilbert's Nullstellensatz\cite{Cox_Little}]\label{thm:Hilb}
 Let $\mathbb{K}$ be an algebraically closed field and $f, f_1,f_2, \cdots,$ $ f_m$ $\in \mathbb{K}[x_1,x_2, \cdots, x_n]$.  $f \in \mathbf{I}(\mathbf{V}(f_1, f_2, \cdots, f_m))$ if, and only if, there exists a positive integer $L$ such that $f^L \in <f_1,f_2, \cdots,f_m>$. 
\end{theorem}
\section{Application of Nullstellensatz to Two-Unicast-Z Networks}\label{sec:cons}
In this section,  we use  Hilbert's Nullstellensatz to describe an  equivalent condition to (\ref{eq:achievability}) for achievability of rate $(1,1)$ in two-unicast-$Z$ networks.  
\begin{corollary}\label{corol}
	The rate $(1,1)$ is not achievable in a $(\mathcal{G},\mathcal{S}_{1},\mathcal{T}_{1},\mathcal{S}_{2},\mathcal{T}_{2})$ two-unicast-$Z$ network with a minimum GNS cut set of size two using scalar linear coding if, and only if, for some $L \in \mathbb{Z}_{+},$ there exists a polynomial  $P$ such that 
	\begin{equation}   \mathbf{G}_{2,1} P\label{eq:ach_1}=\left(\mathbf{G}_{1,1}\mathbf{G}_{2,2}\right)^{L},\end{equation}\label{lem:ach_1}
	where, $\mathbf{G}_{i,j}$ is the transfer polynomial from source $\mathcal{S}_i$ to destination $\mathcal{T}_j$,  $i,j \in\{1,2\}$.
\end{corollary} 
\begin{proof}
First, suppose that there exists a polynomial ${P}$  such that $\mathbf{G}_{2,1} P=\left(\mathbf{G}_{1,1}\mathbf{G}_{2,2}\right)^{L}$, for some positive integer $L$. In order to achieve the rate pair (1,1), we need to satisfy  the conditions in (\ref{eq:achievability}), that is, to set $\mathbf{G}_{2,1}= 0$ such that $\mathbf{G}_{1,1} \neq 0$ and $\mathbf{G}_{2,2} \neq 0$. However, from (\ref{eq:ach_1}), setting  $\mathbf{G}_{2,1}=0$ gives $\mathbf{G}_{1,1} =0$ or $\mathbf{G}_{2,2} =0$. Hence, the first condition in  (\ref{eq:achievability}) cannot be satisfied and rate $(1,1)$ is not achievable in the network using scalar linear coding.

For the other direction, suppose that the rate pair $(1,1)$ is not achievable in the network using scalar linear coding.  Thus the conditions in  (\ref{eq:achievability}) cannot be satisfied simultaneously. Therefore, whenever $\mathbf{G}_{2,1}=0$, we have $\mathbf{G}_{1,1} =0$ or $\mathbf{G}_{2,2} =0$ . In other words, 
\begin{align}\label{eqva}
\mathbf{V}\big(\mathbf{G}_{1,1}\big) \cup \mathbf{V}\big(\mathbf{G}_{2,2}\big) \supseteq \mathbf{V}\big(\mathbf{G}_{2,1}\big).
 \end{align} 
 Since $\mathbf{V}\big(\mathbf{G}_{1,1}\big) \cup \mathbf{V}\big(\mathbf{G}_{2,2}\big)= \mathbf{V}\big(\mathbf{G}_{1,1}\mathbf{G}_{2,2}\big)$ \cite{Cox_Little}, substituting in (\ref{eqva}) gives
 \begin{align}\label{eqvaaa}
\mathbf{V}\big(\mathbf{G}_{1,1}\mathbf{G}_{2,2}\big)\supseteq \mathbf{V}\big(\mathbf{G}_{2,1}\big).
 \end{align} 
 
%Otherwise, if $\mathbf{G}_{2,1}$ can be set to zero  such that neither $\mathbf{G}_{1,1}$ nor $\mathbf{G}_{2,2}$ vanishes, then rate $(1,1)$ is achievable, a contradiction.

Then, from Remark \ref{rev},
 \begin{align}     \label{revI}                                                 
\mathbf{I}\big(\mathbf{V}\big(\mathbf{G}_{1,1}\mathbf{G}_{2,2}\big)\big) \subseteq \mathbf{I}\big(\mathbf{V}\big(\mathbf{G}_{2,1}\big)\big).
 \end{align}
  
 Since $\mathbf{G}_{1,1}\mathbf{G}_{2,2} \in \mathbf{I}\big(\mathbf{V}\big(\mathbf{G}_{1,1}\mathbf{G}_{2,2}\big)\big)$, from (\ref{revI}), we get
 \begin{align}
  \mathbf{G}_{1,1}\mathbf{G}_{2,2}\in\mathbf{I}\big(\mathbf{V}\big(\mathbf{G}_{2,1}\big)\big)
  \end{align}

The last equation satisfies  the hypothesis of Theorem \ref{thm:Hilb} (Hilbert's Nullstellensatz), hence there exists an integer $L\geq 1$ such that $\left(\mathbf{G}_{1,1}\mathbf{G}_{2,2}\right)^L \in  \hspace{2mm}\big<\hspace{0mm}\mathbf{G}_{2,1} \big>$. In other words, there exists a polynomial ${P}$ such that $ \mathbf{G}_{2,1} P=\left(\mathbf{G}_{1,1}\mathbf{G}_{2,2}\right)^{L}$, for some positive integer $L$.
\end{proof}

The contrapositive of Corollary \ref{corol} stated in the next corollary provides the equivalent condition to (\ref{eq:achievability}) for the rate $(1,1)$ achievability in the two-unicast-$Z$ network using scalar linear codes.
\begin{corollary}\label{corolf}
	The rate $(1,1)$ is achievable in a $(\mathcal{G},\mathcal{S}_{1},\mathcal{T}_{1},\mathcal{S}_{2},\mathcal{T}_{2})$ two-unicast-$Z$ network with a minimum GNS cut set of size two using scalar linear coding if, and only if, there  does not exist a polynomial  $P$ such that 
	\begin{equation}   \mathbf{G}_{2,1} P\label{eq:ach_1}=\left(\mathbf{G}_{1,1}\mathbf{G}_{2,2}\right)^{L},\end{equation}\label{lem:ach_1}
	for all  $L \in \mathbb{Z}_{+}$, where $\mathbf{G}_{i,j}$ is the transfer polynomial from source $\mathcal{S}_i$ to destination $\mathcal{T}_j$,  $i,j \in\{1,2\}$.
\end{corollary}

\begin{figure}[]
    \centering
    \includegraphics[scale=0.35]{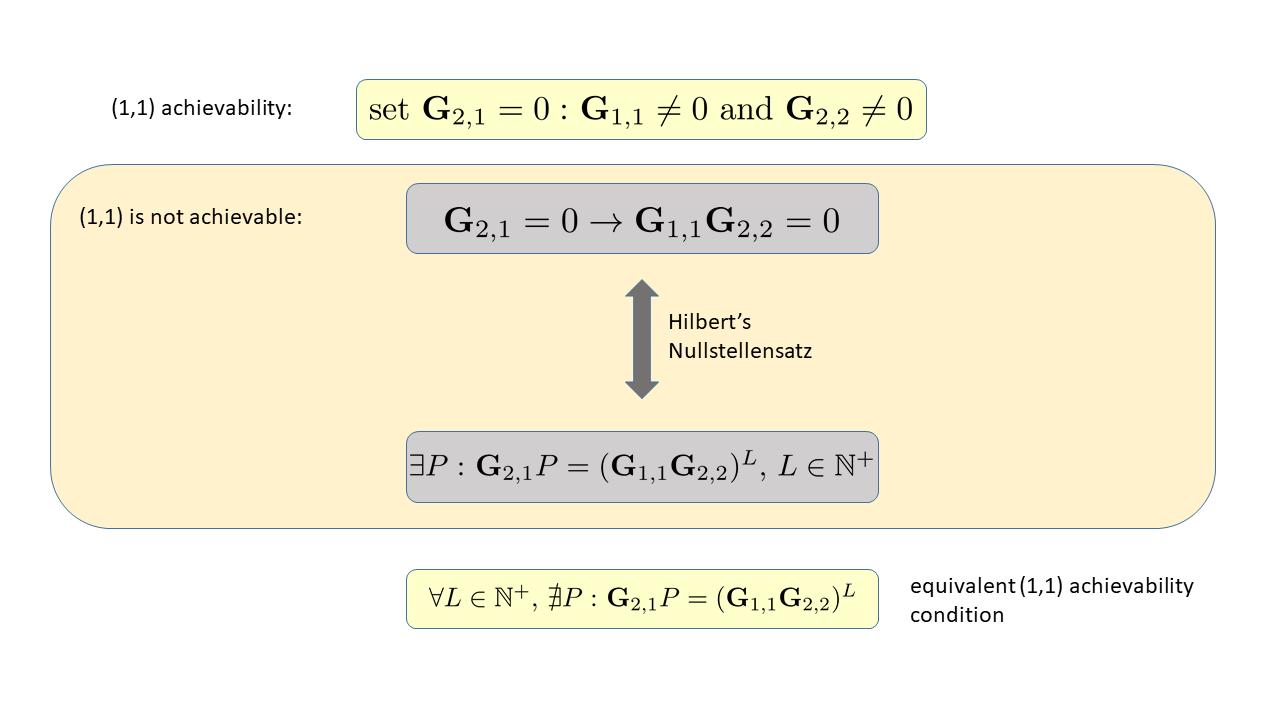}
    \caption{A diagram showing the rate $(1,1)$ commutative algebraic achievability condition for two-unicast-$Z$ networks via scalar linear codes}
    \label{fig:lemproof}
\vspace{-0pt}
\end{figure}
Fig. \ref{fig:lemproof} shows a diagram demonstrating the relations between different achievability conditions for rate $(1,1)$ using scalar linear codes in two-unicast-$Z$ networks.

In the following two sections, we aim to show that the generalized network sharing (GNS) bound is tight in two-unicast-$Z$ networks at rate $(1,1)$. That is, we prove that whenever a two-unicast-$Z$ network has a minimum GNS cut of size two, rate $(1,1)$ is achievable in the network, specifically, using scalar linear codes.  In order to prove that, we first introduce a decomposition of networks, in the next section, that allows representing any network in terms of smaller sub-networks with connections to each other. Using this network decomposition, a variant of the achievability condition in Corollary \ref{corolf} is obtained in Lemma \ref{lem0} and facilitates the achievability proof afterwards. In order to prove the achievability, we use degree arguments to prove the  inexistence of specific polynomials.    

\section{Network Transfer Matrix Decomposition}\label{sec:netdecomp}
In this section, we develop a network  decomposition method that is central to our achievability proof. While our method is more generally applicable, we present our decomposition for the case of a two-unicast-$Z$ network with two GNS edges $\mathcal{C}_{GNS}=\{e_1,e_2\}$ where  $\operatorname{Ord}(e_1) < \operatorname{Ord}(e_2)$. This network decomposition is valid under the scalar linear algebraic framework described in Section \ref{sec:background}, and is formulated in Lemma \ref{lem:netdecomp}. A more general network decomposition for two-unicast-$Z$ networks can be found in Appendix A.  In Section \ref{subsec:cons}, we state and describe a condition, in Lemma \ref{lem0}, that is equivalent to the condition stated in Corollary  \ref{corolf} for achievability of rate $(1,1)$. This  new condition  has  useful properties as will be shown at the end of this section.

In order to understand the motivation behind our network decomposition, 
let $\mathcal{G}=(\mathcal{V},\mathcal{E})$ be a directed acyclic graph and $e_1,e_2 \in \mathcal{E}$. Suppose that all $e_1 \rightarrow e_2$ paths contain some edge $e \in \mathcal{E}$, then $\sum\limits_{p: e_1 \rightarrow e_2} \hspace{-1mm}w(p)$ can be factorized as  $\sum\limits_{p: e_1 \rightarrow e_2} \hspace{-1mm}w(p)= \sum\limits_{p: e_1 \rightarrow e} \hspace{-1mm}w(p) \sum\limits_{p: e \rightarrow e_2} \hspace{-3mm}w(p)$.
This implies that if an edge $e \in \mathcal{E}$ is a single edge cut set in a single unicast network with source edge and destination edge $s,t \in \mathcal{E}$, respectively, then the single unicast network can be decomposed into two concatenated smaller single unicast sub-networks such that one of these smaller sub-networks has $s$ and $e$ as source and destination edges, respectively, and the other    has $e$ and $t$ as source and destination edges, respectively. In addition,  the transfer polynomial of the network, i.e., $\sum\limits_{p: s \rightarrow t} \hspace{-1mm}w(p)$ can be written in terms of the transfer polynomials of the  sub-networks, i.e., $\sum\limits_{p: s \rightarrow e} \hspace{-1mm}w(p)$ and $\sum\limits_{p: e \rightarrow t} \hspace{-1mm}w(p)$, in  the following product form, 
\begin{equation}
    \sum\limits_{p: s \rightarrow t} \hspace{-1mm}w(p)=\sum\limits_{p: s \rightarrow e} \hspace{-1mm}w(p)\sum\limits_{p: e \rightarrow t} \hspace{-1mm}w(p).
\end{equation}
Fig. \ref{fig:singdecomp} illustrates the decomposition of a single unicast network with respect to  the single edge cut set $e$, and the corresponding relation between the transfer polynomials of the original network and the resultant sub-networks.

\begin{figure}[]
\vspace{5mm}
    \centering
    \includegraphics[scale=0.5]{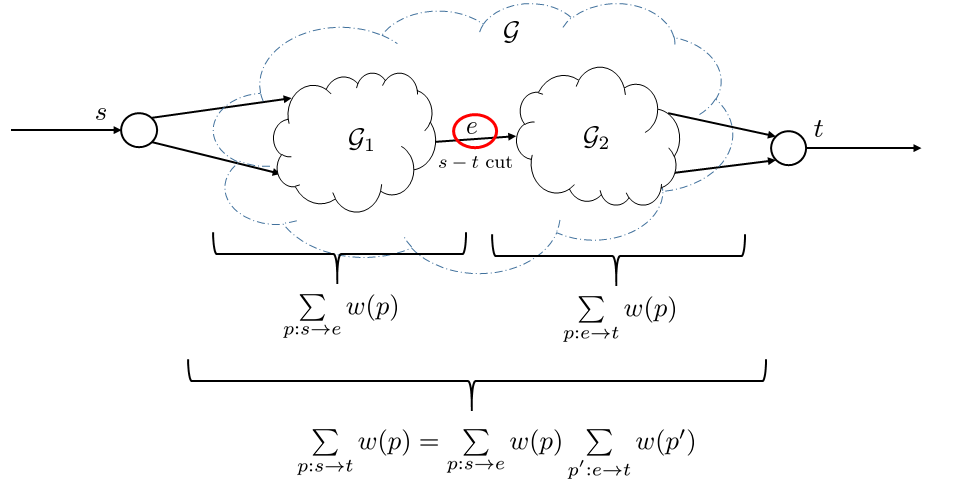}
    \caption{A  single unicast network decomposed into two single unicast sub-networks with respect to the single edge cut set $e$ and the corresponding relation between the transfer polynomials of the different networks. }
    \label{fig:singdecomp}
\vspace{0pt}
\end{figure}

The network decomposition lemma  presented in this section aims to generalize  the idea of decomposing the single unicast network with respect to any single edge cut set in the network to the two-unicast-$Z$ networks. In particular, the network decomposition lemma describes the decomposition of any two-unicast-$Z$ with respect to any  GNS cut set of size two in the network. The main challenge is that if the edges of the GNS cut communicate with each other, then you will have to carefully untangle the effect of the weight of this path that goes through both edges. Before we present the network decomposition lemma for two-unicast-$Z$ networks, we first provide some definitions.

Let $(\mathcal{G},\mathcal{S}_{1},\mathcal{T}_{1},\mathcal{S}_{2},\mathcal{T}_{2})$ be a two-unicast-$Z$ network containing a GNS cut set $\mathcal{C}_{GNS}=\{e_1,e_2\}$ of size two. This two-unicast-$Z$ network can be decomposed into two two-unicast-$Z$ networks: The first one, denoted \emph{the left-side network}, is a two-unicast-$Z$ network with $\mathcal{S}_1,\mathcal{S}_2$ as sources and $e_1,e_2$ as destinations, and the second one, denoted \emph{the right-side network}, is a two-unicast-$Z$ network with $e_1,e_2$ as sources and $\mathcal{T}_1,\mathcal{T}_2$ as destinations. The left-side and right-side networks are formally described in the following two definitions.
\begin{definition}[Left-side  network]
%Consider a $(\mathcal{G},\mathcal{S}_1,\mathcal{T}_1,\mathcal{S}_2, \mathcal{T}_2)$ two-unicast-$Z$ network  with GNS cut set  $\mathcal{C}_{GNS}=\{e_1,e_2\}$,  $\operatorname{Ord}(e_1) < \operatorname{Ord}(e_2)$.   Let $\mathcal{P}_1$ a set of paths defined as $p \in \mathcal{P}_1$ if $p$ is an $s \rightarrow e$ path for some $s \in \mathcal{S}_1 \cup \mathcal{S}_2$ and $e \in \mathcal{C}_{GNS}$. Let $\mathcal{V}_1$ be the set of vertices such that $v \in \mathcal{V}_{1}$ if and only if $v$ belongs to some path in $\mathcal{P}_1$. $\mathcal{E}_1$ be the set of edges  such that $e \in \mathcal{E}_1$ if and only if $e$ belongs to some path in $\mathcal{P}_1$. The \emph{left-side} side of a two-unicast-$Z$ network is defined to be the subgraph $\mathcal{G}_1=(\mathcal{V}_1,\mathcal{E}_1)$. 
Consider a $(\mathcal{G},\mathcal{S}_1,\mathcal{T}_1,\mathcal{S}_2, \mathcal{T}_2)$ two-unicast-$Z$ network  with GNS cut set  $\mathcal{C}_{GNS}=\{e_1,e_2\}$,  $\operatorname{Ord}(e_1) < \operatorname{Ord}(e_2)$.    %The \emph{left-side} side of the  network is defined as the subgraph $\mathcal{G}_1=(\mathcal{V}_1,\mathcal{E}_1) \subseteq \mathcal{G}$, where $\mathcal{V}_1=\{v \in \mathcal{V}: v \text{ belongs to some }\mathcal{S}_1 \cup \mathcal{S}_2 \rightarrow \mathcal{C}_{GNS} \text{ path }\}$ and $\mathcal{E}_1=\mathcal{S}_1 \cup \mathcal{S}_2 \cup \mathcal{C}_{GNS}\cup \mathcal{E}'_1$, where $\mathcal{E}'_1=\{e=(v_1,v_2) \in \mathcal{E}: v_1,v_2 \in \mathcal{V}_1\}$.
The \emph{left-side}    network is defined as the subgraph $\mathcal{G}_1=(\mathcal{V}_1,\mathcal{E}_1) \subseteq \mathcal{G}$, where $\mathcal{E}_1=\{e \in \mathcal{E}: e \text{ belongs to some }\mathcal{S}_1 \cup \mathcal{S}_2 \rightarrow \mathcal{C}_{GNS} \text{ path}\}$ and $\mathcal{V}_1=\{v \in \mathcal{V}: v \text{ is the head or tail of edge $e$, for some $e \in \mathcal{E}_{1}$}\}$. 
\end{definition}
\begin{definition}[Right-side network]
Consider a $(\mathcal{G},\mathcal{S}_1,\mathcal{T}_1,\mathcal{S}_2, \mathcal{T}_2)$ two-unicast-$Z$ network  with GNS cut set  $\mathcal{C}_{GNS}=\{e_1,e_2\}$,  $\operatorname{Ord}(e_1) < \operatorname{Ord}(e_2)$.  The \emph{right-side}   network is defined as the subgraph $\mathcal{G}_2=(\mathcal{V}_2,\mathcal{E}_2) \subseteq \mathcal{G}$, where $\mathcal{E}_2=\{e \in \mathcal{E}: e \text{ belongs to some }\mathcal{C}_{GNS} \rightarrow \mathcal{T}_1 \cup \mathcal{T}_2 \text{ path}\}$ and  $\mathcal{V}_2=\{v \in \mathcal{V}: v \text{ is the head or tail of edge $e$, for some $e \in \mathcal{E}_{2}$}\}$. 
\end{definition}

%\begin{remark}\label{rmk2}
%Let  $s_1,s_2,t_1,t_2$ be edges, $\sum\limits_{\substack{p_1: s_1 \rightarrow t_1 \\  p_2:s_2 \rightarrow t_2}} \hspace{-4mm}w(p_1) w(p_2) = \sum\limits_{p: s_1 \rightarrow t_1} \hspace{-3mm}w(p) \sum\limits_{p: s_2 \rightarrow t_2} \hspace{-3mm}w(p) $.
%\end{remark}

Now, we have the following definitions.
\begin{definition}[Transfer matrix]
Consider a DAG $\mathcal{G}=(\mathcal{V},\mathcal{E})$, let $\mathcal{S}'=\{s'_1,s'_2, \cdots, s'_m\}$ and $\mathcal{T}'=\{t'_1,t'_2, \cdots, t'_n\}$ be any two subsets of $\mathcal{E}$. The transfer matrix $\mathbf{M}_{(\mathcal{S}',\mathcal{T}')}$ is defined as the $m \times n$  matrix whose entry at the index  $(i,j)$ is $\sum\limits_{p: s'_{i} \rightarrow t'_{j}} \hspace{-3mm}w(p)$.
\end{definition}
Note that $\mathbf{M}_{(\mathcal{S}',\mathcal{T}')}$ is the submatrix of the network extended transfer matrix $\mathbf{H}$ with rows (columns) corresponding to $\mathcal{S}'$ $(\mathcal{T}')$. Now, we need to define a specific transfer matrix that describes the received outputs on the destination edges with respect to the source messages. In order to do so, we define the network transfer matrix. Similar to the transfer matrix, the network transfer matrix is a submatrix of the network extended transfer matrix $\mathbf{H}$. However, for the network transfer matrix, the rows are only associated to the source edges, i.e., $\mathcal{S}_1\cup\mathcal{S}_2$, and the columns are only associated to the destination edges, i.e., $\mathcal{T}_1\cup\mathcal{T}_2$. This this different from the transfer matrices  where  rows and columns can correspond to any arbitrary set of edges in the network. A formal definition for the network transfer matrix can be as follows.
\begin{definition}[Network transfer matrix]
 Consider a $(\mathcal{G},\mathcal{S}_1,\mathcal{T}_1,\mathcal{S}_2, \mathcal{T}_2)$ two-unicast-$Z$ network. The network transfer matrix is  $\mathbf{M}_{(\mathcal{S}_1\cup\mathcal{S}_2,\mathcal{T}_1\cup\mathcal{T}_2)}$.
\end{definition}

 We also define another special case of a transfer matrix that captures the underlying algebraic properties among any subset of edges in the network. This type of transfer matrix is denoted by the coupling matrix and is formally defined as follows. 

\begin{definition}[Coupling matrix]
Consider a DAG $\mathcal{G}=(\mathcal{V},\mathcal{E})$, let $\mathcal{U}=\{u_1,u_2, \cdots, u_m\}$ be any subset of $\mathcal{E}$, let the edges of the set be ordered such that $\Ord(u_j) > \Ord(u_i)$ if $j>i$. The coupling matrix $\mathbf{\Lambda}^{\mathcal{U}}$ is defined as the $m \times m$ upper triangular matrix  whose $(i,j)$-th entry is $\sum\limits_{p: u_{i} \rightarrow u_{j}} \hspace{-3mm}w(p)$ if $i < j$, or one if $i=j$. 
\end{definition}

%\begin{definition}[Path system]
%Let $\mathcal{S}=\{s_1, s_2, \cdots, s_n\}$ and $\mathcal{T}=\{t_1, t_2, \cdots, t_n\}$ be two sets of $n$ vertices, where $\mathcal{S}$ and $\mathcal{T}$ may not be disjoint, and let $\sigma$ be some permutation, we define a path system $\mathcal{P}$ with a n-tuple of paths as $\mathcal{P}=(p_1, p_2, \cdots, p_n)$, where $p_i : s_i \rightarrow t_{\sigma(i)}$, $i \in [n]$. 
%\end{definition}

\begin{figure*}%{!t}
$$
\mathbf{M}= \left(\begin{array}{ccc} 
\sum\limits_{p: s_{1} \rightarrow t_{1}} \hspace{-0mm}w(p) & \sum\limits_{p: s_{1} \rightarrow t_{2} \hspace {1mm} \text{via} \hspace {1mm}\mathcal{C}_{GNS}} \hspace{-3mm}w(p)   \\
\sum\limits_{p: s_{2} \rightarrow t_{1}} \hspace{-0mm}w(p) & \sum\limits_{p: s_{2} \rightarrow t_{2}} \hspace{-0mm}w(p)
\end{array}\right), \mathbf{\Lambda}= \left(\begin{array}{ccc} 
1 & \sum\limits_{p: e_{1} \rightarrow e_{2} } \hspace{-0mm}w(p) \\
0 & 1  
\end{array}\right),
$$
\begin{align}\label{eq:Ms}
\mathbf{M}_1= \left(\begin{array}{ccc} 
\sum\limits_{p: s_{1} \rightarrow e_{1}} \hspace{-0mm}w(p) & \sum\limits_{p: s_{1} \rightarrow e_{2} \backslash e_{1}} \hspace{-3mm}w(p)  \\
\sum\limits_{p: s_{2} \rightarrow e_{1}} \hspace{-0mm}w(p) & \sum\limits_{p: s_{2} \rightarrow e_{2}\backslash e_{1}} \hspace{-3mm}w(p) 
\end{array}\right), 
\mathbf{M}_2= \left(\begin{array}{ccc} 
\hspace{-0mm}\sum\limits_{p: e_{1} \rightarrow t_{1} \backslash {e_{2}}} \hspace{-3mm}w(p) & \sum\limits_{p: e_{1} \rightarrow t_{2} \backslash {e_{2}}} \hspace{-3mm}w(p)  \\
\sum\limits_{p: e_{2} \rightarrow t_{1}} \hspace{-0mm}w(p) & \sum\limits_{p: e_{2} \rightarrow t_{2}} \hspace{-0mm}w(p)  \\
\end{array}\right).
\end{align}
\hrulefill
\end{figure*}
Now, given the above definitions and 
recalling that $\mathbf{F}$ is the local coding matrix of the network and   $\bar{\mathbf{F}}$ is the set whose elements are the (non-zero) entries of $\mathbf{F}$, we propose the following lemma on network decomposition. 
\begin{lemma}[Network Decomposition Lemma]
\label{lem:netdecomp}
Consider a $(\mathcal{G},\mathcal{S}_1,\mathcal{T}_1,\mathcal{S}_2, \mathcal{T}_2)$ two-unicast-$Z$ network  with GNS cut set  $\mathcal{C}_{GNS}=\{e_1,e_2\}$,  $\operatorname{Ord}(e_1) < \operatorname{Ord}(e_2)$. Let $\mathcal{G}_1$ and $\mathcal{G}_2$ be the graphs of the left-side network and right-side  network, respectively, with respect to $\mathcal{C}_{GNS}$. Then 
\begin{enumerate}[label=(\alph*)]
\item $\mathbf{M}=\mathbf{M_1}\mathbf{\Lambda}\mathbf{M_2}$, where explicit expressions for matrices $\mathbf{M},\mathbf{M}_{1},\mathbf{M}_{2},\mathbf{\Lambda}$ are shown at the top of this page.
	\item In graph $\mathcal{G}_{1}$, the network transfer matrix from the source $\{s_{1},s_2\}$ to edges $\{e_1,e_2\}$ is $\mathbf{M}_{1} \mathbf{\Lambda}.$ In graph $\mathcal{G}_{2}$, the network transfer matrix from $\{e_{1},e_{2}\}$ to $\{t_{1},t_{2}\}$ is $\mathbf{\Lambda}\mathbf{M}_{2}$ 
	\item Let $\mathbf{F}_{1} \subset \bar{\mathbf{F}}$ be the set of variables in  $ \sum\limits_{p: s_i \rightarrow e_1} \hspace{-3mm} w(p) )$, $i \in \{1,2\}$, and let $\mathbf{F}_2 \subset \bar{\mathbf{F}}$ be the set of variables in $\mathbf{M}_{2}$, we have $\mathbf{F}_{1} \cap \mathbf{F}_{2} = \phi$. 
	
	\item  Let $\mathbf{F}_{1} \subset \bar{\mathbf{F}}$ be the set of variables in  $ \mathbf{M}_1$, and let $\mathbf{F}_2 \subset \bar{\mathbf{F}}$ be the set of variables in $\sum\limits_{p: e_2 \rightarrow t_i} \hspace{-3mm} w(p)$, $i \in \{1,2\}$, we have $\mathbf{F}_{1} \cap \mathbf{F}_{2} = \phi$. 
	
	\item  Let $\mathbf{F}_{1} \subset \bar{\mathbf{F}}$ be the set of variables in  $ \mathbf{M}_1$, and let $\mathbf{F}_2 \subset \bar{\mathbf{F}}$ be the set of variables in $\mathbf{M}_2$. If there are no $e_1 \rightarrow e_2$ paths in $\mathcal{G}$, then $\mathbf{F}_{1} \cap \mathbf{F}_{2} = \phi$. 
	%\item For any two subsets $\mathbf{F}_{1},\mathbf{F}_{2} \subset \mathbf{F}$ such that $ \sum\limits_{p: s_i \rightarrow e_1} \hspace{-3mm} w(p)  \subseteq \mathbb{K}(\mathbf{F}_{1})$, $i \in \{1,2\}$ and $\mathbf{M}_{2} \subseteq \mathbb{K}(\mathbf{F}_{2})$, we have $\mathbf{F}_{1} \cap \mathbf{F}_{2} = \phi$. 
	
	%\item  For any two subsets $\mathbf{F}_{1},\mathbf{F}_{2} \subset \mathbf{F}$ such that $\mathbf{M}_{1} \subseteq \mathbb{K}(\mathbf{F}_{1})$ and $ \sum\limits_{p: e_2 \rightarrow t_i} \hspace{-3mm} w(p)  \subseteq \mathbb{K}(\mathbf{F}_{2})$, $i \in \{1,2\}$, we have $\mathbf{F}_{1} \cap \mathbf{F}_{2} = \phi$. 
	
	%\item  If there are no $e_1 \rightarrow e_2$ paths, then, for any two subsets $\mathbf{F}_{1},\mathbf{F}_{2} \subset \mathbf{F}$ such that $\mathbf{M}_{1} \subseteq \mathbb{K}(\mathbf{F}_{1})$ and $\mathbf{M}_2  \subseteq \mathbb{K}(\mathbf{F}_{2})$, $i \in \{1,2\}$, we have $\mathbf{F}_{1} \cap \mathbf{F}_{2} = \phi$. 
	\end{enumerate}
\end{lemma}
The proof of this lemma is in Appendix B. An illustration of the decomposed network and the resultant sub-networks with their specified network transfer matrices are shown in Fig. \ref{fig:decompo}.    
\begin{figure}[]
\vspace{5mm}
    \centering
    \includegraphics[scale=0.5]{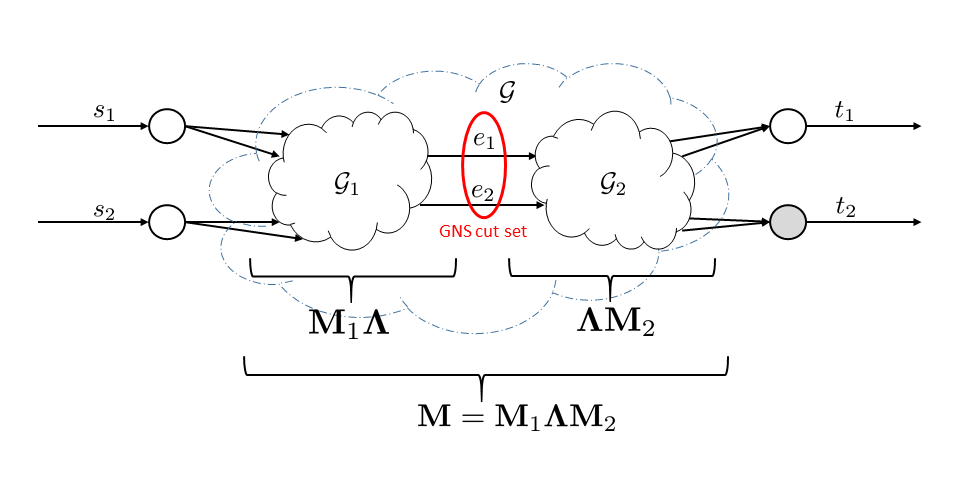}
    \caption{A two-unicast-$Z$ network decomposition into left-side and right-side networks and the corresponding network transfer matrices. }
    \label{fig:decompo}
\vspace{0pt}
\end{figure}
\subsection{Consequences of Network Decomposition}\label{subsec:cons}
In the following, we  describe an equivalent condition to the condition stated in Corollary  \ref{corolf} for achievability of rate $(1,1)$. This equivalent condition is stated in Lemma \ref{lem0} and makes advantage of the network decomposition lemma in order to get some favorable properties that will be discussed at the end of this section and will help in developing the rate $(1,1)$ feasibility proofs in Section \ref{sec:intermediate}.
\begin{lemma}\label{lem0}
The rate $(1,1)$ is  achievable in a $(\mathcal{G},\mathcal{S}_{1},\mathcal{T}_{1},\mathcal{S}_{2},\mathcal{T}_{2})$ two-unicast-$Z$ network with a minimum GNS cut set of size two using scalar linear coding if, and only if,  there does not exist  a polynomial $P$ such that 
\begin{equation}
P \sum\limits_{p: s_2 \rightarrow t_{1}}\hspace{-0mm}w(p) = \big(\det(\mathbf{M})\big)^{L}, \label{eq:ach_2}
\end{equation}
for all  $L \in \mathbb{Z}_{+},$ where $\mathbf{M}$ is as defined in (\ref{eq:Ms}).
\end{lemma}  
\begin{proof}

The proof of Lemma \ref{lem0} follows from Corollary \ref{corolf} by noting that $\sum\limits_{p: s_{1} \rightarrow t_{1}} \hspace{-0mm}w(p)=\mathbf{G}_{1,1}$, $\sum\limits_{p: s_{2} \rightarrow t_{2}} \hspace{-0mm}w(p)=\mathbf{G}_{2,2}$ and $\sum\limits_{p: s_{2} \rightarrow t_{1}} \hspace{-0mm}w(p)=\mathbf{G}_{2,1}$, in addition to the fact that $\big(\det(\mathbf{M})\big)^{L}$ can be written as $\big(\det(\mathbf{M})\big)^{L}= (\mathbf{G}_{1,1}\mathbf{G}_{2,2})^L+ P_0\mathbf{G}_{2,1}$ for some polynomial $P_0$.
\end{proof}
\begin{corollary}\label{cor:Mn0} Let $(\mathcal{G},\mathcal{S}_{1},\mathcal{T}_{1},\mathcal{S}_{2},\mathcal{T}_{2})$ be a two-unicast-$Z$ network with a minimum GNS cut set of size two, $ \det( \mathbf{M}) \neq 0$ where $\mathbf{M}$ is as defined in (\ref{eq:Ms}). 
\end{corollary}
\begin{proof}
Let the minimum GNS cut set of size two in the network be $\{e_1,e_2\}$. For the sake of contradiction, suppose  that $\det( \mathbf{M})=0$. Thus, by part (a) of Lemma \ref{lem:netdecomp}, $\det( \mathbf{M}_i)=0$, $i =1 \text{ or } 2$. 
 Therefore, by max-flow min-cut theorem, there is a  single edge cut set  in $(\mathcal{G}_1, \mathcal{S}_1, e_1, \mathcal{S}_2, e_2)$ or $(\mathcal{G}_2, e_1, \mathcal{T}_1, e_2, \mathcal{T}_2)$.  However, a single edge cut set  in $(\mathcal{G}_1, \mathcal{S}_1, e_1, \mathcal{S}_2, e_2)$ or $(\mathcal{G}_2, e_1, \mathcal{T}_1, e_2, \mathcal{T}_2)$ is a  single edge GNS cut set  in $(\mathcal{G},\mathcal{S}_1, \mathcal{T}_1, \mathcal{S}_2,\mathcal{T}_2)$, a contradiction to the fact that  the network $(\mathcal{G},\mathcal{S}_1, \mathcal{T}_1, \mathcal{S}_2,\mathcal{T}_2)$ has a minimum GNS cut set of size two. Hence, $ \det( \mathbf{M}) \neq 0$.
\end{proof}
\subsection{Notations and Observations}
In the following, we introduce some notations which will be used for the rest of the paper. In addition, based on these  notations, we give some observations on the advantage of the achievability condition derived in Lemma \ref{lem0}.

Recalling that $\mathcal{C}_{GNS}=\{e_1,e_2\}$ is a GNS cut set in our two-unicast-$Z$ network, for $i \in \{1,2\}$, let $\mathbf{u}_{i}$ denote an $\In(e_i) \times 1$ vector of indeterminate variables representing the local coding coefficients from the edges incoming into $e_i$ to $e_i$. Specifically, denoting $\In(e_i)=\{e_{i,1},e_{i,2},\ldots, e_{i,|\In(e_i)|}\}$, the vector $\mathbf{u}_{i}$ is equal to $(\beta_{e_{i,1},e_{i}}, \beta_{e_{i,2},e_{i}}, \ldots, \beta_{e_{i,|\In(e_{i}|},e_i})$. We now aim to express the polynomials in $\mathbf{M},\mathbf{M}_{1},\mathbf{\Lambda}$ as polynomials in $\mathbf{u}_{1},\mathbf{u}_{2}$. We write
\begin{align}\label{eq:M1}
\mathbf{M}_1= \left( \begin{array}{ccc}
	\mathbf{a}_1\mathbf{u}_{1} & \mathbf{a}_2\mathbf{u}_2\\
	\mathbf{b}_1\mathbf{u}_{1} & \mathbf{b}_2\mathbf{u}_2 \end{array} \right), ~~
\mathbf{\Lambda}=\left(\begin{array}{ccc} 1 & \mathbf{\lambda}_{12}  \mathbf{u}_2 \\ 0&1\\ \end{array}\right),
\end{align}
where $\mathbf{a}_i$, with $i \in \{1,2\}$, is the  $1 \times \In(e_i) $ vector of transfer polynomials from $s_{1}$ to $\In(e_{i})$ containing paths that do not go through $e_j, j\neq i$. Specifically, $\mathbf{a}_{i}= (a_{i,1}, a_{i,2}, \cdots, a_{i,|\In(e_i)|})$, where $a_{i,j} = \sum\limits_{p:s_{1} \rightarrow e_{i,j} \backslash \{e_{k}:k\neq i\}} \hspace{-9mm} w(p)$, $j \in \{1, \cdots, |\In(e_i)|\}$. The row vectors $\mathbf{b_1}$ and  $\mathbf{b_2}$  are defined similarly but with respect to $s_{2}$. Let  $\mathbf{\lambda}_{12}= (\lambda_{12,1}, \lambda_{12,2}, \cdots, \lambda_{12,|\In(e_2)|})$ be a $1 \times  \In(e_2)$ vector where  $ \lambda_{12,|\In(e_2)|} = \sum\limits_{p:e_{1} \rightarrow e_{2,j} } \hspace{-0mm} w(p)$, $j \in \{1, \cdots, |\In(e_2)|\}$. Finally, we write
\begin{align}\label{eq:M2}
\mathbf{M}_2= \left( \begin{array}{ccc}
	\mu_{11} & \mu_{12}\\
	\mu_{21} & \mu_{22} \end{array} \right),
\end{align}
where  $\mu_{ij}=\sum\limits_{p:e_{i} \rightarrow t_{j} \backslash \{e_{k}: k\neq i\}} \hspace{-9mm} w(p)$, $i,j \in \{1,2\}$.

 Now, recalling that $\mathbf{M}=\mathbf{M}_1\mathbf{\Lambda}\mathbf{M}_2$ (Lemma \ref{lem:netdecomp}) where $\det(\mathbf{\Lambda})=1$, we have 
 \begin{align}
 \det(\mathbf{M})&= \det(\mathbf{M}_1)\det(\mathbf{M}_2) \notag\\
 &=  \big( \mathbf{a}_1\mathbf{u}_1 \hspace{1mm}\mathbf{b}_2\mathbf{u}_2 - \mathbf{b}_1\mathbf{u}_1\hspace{1mm} \mathbf{a}_2\mathbf{u}_2\big)\big(\mu_{11}\mu_{22}-\mu_{12}\mu_{21}\big).\notag
 \end{align} 
 Moreover, $\mathbf{M}=\mathbf{M}_1\mathbf{\Lambda}\mathbf{M}_2$ also implies that
 \begin{align}\label{eq:s2t1}
 \sum\limits_{p: s_2 \rightarrow t_{1}}\hspace{-0mm}w(p)&= \sum\limits_{p: s_2 \rightarrow e_{1} }\hspace{-0mm}w(p) \big(\sum\limits_{p: e_1 \rightarrow t_{1}\backslash e_2}\hspace{-0mm}w(p)+\sum\limits_{p: e_1 \rightarrow e_{2}}\hspace{-0mm}w(p)\sum\limits_{p: e_2 \rightarrow t_{1}}\hspace{-0mm}w(p)\big)+\sum\limits_{p: s_2 \rightarrow e_{2}\backslash e_1}\hspace{-0mm}w(p)\sum\limits_{p: e_2 \rightarrow t_{1}}\hspace{-0mm}w(p)\notag\\
 &=\mathbf{b}_1\mathbf{u}_1 (\mu_{11}+\mathbf{\lambda}_{12}\mathbf{u}_{2} \mu_{21})+ \mathbf{b}_2\mathbf{u}_2 \mu_{21}.
 \end{align}
 
 Therefore, (\ref{eq:ach_2}) can be written as  
\begin{equation}\label{eq:ach_3}
\big(\mathbf{b}_1\mathbf{u}_1 (\mu_{11}+\mathbf{\lambda}_{12}\mathbf{u}_{2} \mu_{21})+ \mathbf{b}_2\mathbf{u}_2 \mu_{21}\big){P}= \big(\mu_{11}\mu_{22}-\mu_{12}\mu_{21}\big)^L\big( \mathbf{a}_1\mathbf{u}_1 \hspace{1mm}\mathbf{b}_2\mathbf{u}_2 - \mathbf{b}_1\mathbf{u}_1\hspace{1mm} \mathbf{a}_2\mathbf{u}_2\big)^L.
\end{equation}

The main utility of Lemma \ref{lem0} is that it ``homogenizes'' the right hand side of Corollary \ref{lem:ach_1} with respect to variables $\mathbf{u}_{i},i=1,2$. To see this more clearly, we state some basic definitions related to the degree of multi-variate polynomials and orderings on monomials.
\subsubsection*{\textbf{Background on orderings on monomials}}

We introduce a brief background on orderings on monomials \cite{Cox_Little}. %A more extended background can be found in Appendix D.
\begin{definition}[Multi-degree of a monomial]
For any monomial $m= x_1^{\alpha_1} x_2^{\alpha_2} \cdots  x_n^{\alpha_n}$  in the polynomial ring $\mathbb{K}[x_1, \cdots, x_n]$, the multi-degree of this monomial is  $\underset{\mathbb{K}[x_1, \cdots, x_n]}{\operatorname{multideg}}(m)=(\alpha_1, \alpha_2, \cdots, \alpha_n) \in \mathbb{Z}^n_{\geq 0}$.
\end{definition}
\begin{definition}[Sum-degree of a monomial]
For any monomial $m= x_1^{\alpha_1} x_2^{\alpha_2} \cdots  x_n^{\alpha_n}$ in the polynomial ring $\mathbb{K}[x_1, \cdots, x_n]$ with $\underset{\mathbb{K}[x_1, \cdots, x_n]}{\operatorname{multideg}}(m)=(\alpha_1, \alpha_2, \cdots, \alpha_n)$, the sum-degree of this monomial is $\underset{\mathbb{K}[x_1, \cdots, x_n]}{\operatorname{sumdeg}}(m)=\sum\limits_{i=1}^n \alpha_i $.
\end{definition}
\begin{definition}[Monomial ordering]
A monomial ordering in $\mathbb{K}[x_1,\cdots, x_n]$ is any relation $>$ on the set of monomials $\mathcal{M}=\{x_1^{\alpha_1} x_2^{\alpha_2} \cdots x_n^{\alpha_n}:(\alpha_1, \alpha_2, \cdots, \alpha_n) \in \mathbb{Z}^n_{\geq 0}\}$ such that:
\begin{enumerate}[label=(\alph*)]
    \item $>$ is a total ordering on $\mathcal{M}$.
    \item $>$ respects multiplication. That is, for any $m_1,m_2,m_3 \in \mathcal{M}$, if $m_1 > m_2$, then $m_1m_3 > m_2m_3$.
    \item $>$ is a well ordering. That is, every nonempty subset of $\mathcal{M}$ has a smallest element under $>$.
\end{enumerate}
\end{definition}
\begin{definition}[Multi-degree of a polynomial]
Let $p=\sum_i^N a_i m_i$ be a nonzero polynomial in the polynomial ring $\mathbb{K}[x_1, \cdots, x_n]$  where, for all $i \in \{1, \cdots, N\}$, $a_i \in \mathbb{K}$ and $m_i$ is a monomial in $\mathbb{K}[x_1, \cdots, x_n]$ and let $>$ be a monomial order. Then, the multi-degree of $p$ is $$\underset{\mathbb{K}[x_1, \cdots, x_n]}{\operatorname{multideg}}(p)=\underset{\mathbb{K}[x_1, \cdots, x_n]}{\operatorname{multideg}}\big(\underset{{i\in \{1,\cdots,N\}}}{\max} (m_i)\big),$$ where the maximum is taken with respect to $>$. 
\end{definition}

\begin{definition}[Sum-degree of a polynomial]
Let $p=\sum_i^N a_i m_i$ be a nonzero polynomial in the polynomial ring $\mathbb{K}[x_1, \cdots, x_n]$ where, for all $i \in \{1, \cdots, N\}$, $a_i \in \mathbb{K}$ and $m_i$ is a monomial in $\mathbb{K}[x_1, \cdots, x_n]$ and let $>$ be a monomial order. Then, the sum-degree of $p$ is $$\underset{\mathbb{K}[x_1, \cdots, x_n]}{\operatorname{sumdeg}}(p)=\underset{\mathbb{K}[x_1, \cdots, x_n]}{\operatorname{sumdeg}}\big( \underset{{i\in \{1, \cdots, N\}}}{\max}(m_i)\big),$$ where the maximum is taken with respect to $>$.
\end{definition}

\begin{definition}[Homogeneous polynomials]
A polynomial $p$ in the polynomial ring $\mathbb{K}[x_1, \cdots, x_n]$ is  homogeneous of sum-degree $s$ if every monomial in $p$ has sum-degree $s$. 
\end{definition}

\begin{lemma}\label{rmk:homog}
Let $p,g,h$ be non-zero polynomials in the polynomial ring $\mathbb{K}[x_1, \cdots, x_n]$ such that $p=gh$. If $p$ is homogeneous, then $g$ and $h$ are also homogeneous.
\end{lemma}

The proof of this lemma can be found in \cite[Chapter~7]{Cox_Little}.

%\begin{lemma}\label{lem:powerL}
%Let $p=\sum_i^N a_i m_i$ be a nonzero polynomial in the polynomial ring $\mathbb{K}[x_1, \cdots, x_n]$ where, for all $i \in \{1, \cdots, N\}$, $a_i \in \mathbb{K}$ and $m_i$ is a monomial in $\mathbb{K}[x_1, \cdots, x_n]$ such that $\underset{\mathbb{K}[x_1, \cdots, x_n]}{\operatorname{sumdeg}}(m_i)=1$. Then,  $\underset{\mathbb{K}[x_1, \cdots, x_n]}{\operatorname{sumdeg}}(p^L)=L$ for any integer $L\geq 0$.
%\end{lemma}

%\begin{remark}\label{rmk:polyadd}
%Let $p_1,p_2$ be nonzero polynomials in the polynomial ring $\mathbb{K}[x_1, \cdots, x_n]$. Then,  $$\underset{\mathbb{K}[x_1, \cdots, x_n]}{\operatorname{sumdeg}}(p_1p_2)=\underset{\mathbb{K}[x_1, \cdots, x_n]}{\operatorname{sumdeg}}(p_1)+\underset{\mathbb{K}[x_1, \cdots, x_n]}{\operatorname{sumdeg}}(p_2).$$
%\end{remark}

\vspace{5pt}
\subsubsection*{\textbf{Observations}}
 For any field $\mathbb{K}$ and any set of indeterminates $x_1, x_2, \ldots, x_n,$ we denote the field of fractions containing the polynomial ring $\mathbb{K}[x_1, x_2,\ldots, x_n]$ as $\mathbb{K}(x_1, x_2, \ldots, x_n)$. Let us denote by $\overline{\mathbb{K}},$ the field of fractions $\mathbb{K}(\mathbf{F}-\{\mathbf{u}_{1}, \mathbf{u}_2\}).$ For $i \in \{1,2\}$, we will also denote by $\overline{\mathbb{K}}^{(i)}$, the polynomial ring $\overline{\mathbb{K}}(\mathbf{u}_{j})[\mathbf{u}_{i}]$ where ${j} \in \{1,2\}-\{i\}$. In $\overline{\mathbb{K}}^{(i)}$, the elements of $\mathbf{u}_i$ define the variables and the coefficients are rational functions of $\mathbf{u}_j$.
 
 Notice that for a network coding coefficient polynomial $P$,  the quantity ${\operatorname{sumdeg}}_{\overline{\mathbb{K}}^{(i)}}(P)$ represents the sum-degree of polynomial $P$ with respect to the indeterminates in $\mathbf{u}_{i}$ alone. Based on this notation, we can make the following important observation: For every monomial $m$ in $\det(\mathbf{M})^L$, we have
$  {\operatorname{sumdeg}}_{{\overline{\mathbb{K}}^{(i)}}}(m)= L, i=1,2.$ That is, the polynomial $\det(\mathbf{M})^L$ is homogeneous  of sum-degree $L$ in $\overline{\mathbb{K}}^{(1)}$ and $\overline{\mathbb{K}}^{(2)}$.
In effect, the above equation means that every monomial on the left hand side of (\ref{eq:ach_2}) of Lemma \ref{lem0} should also have a sum-degree of $L$ with respect to the variables in $\mathbf{u}_{i}$ alone, for each $i=1,2.$ Notice that, in contrast, the right hand side of  Corollary \ref{lem:ach_1} does not necessarily satisfy this property. Lemma \ref{lem0} will be used to show Theorem \ref{bathm}. In particular, we will show that if the graph in a two-unicast-$Z$ network satisfies certain properties, then it is not possible to find polynomial $P$ satisfying (\ref{eq:ach_2}). %because  the degree of the left hand side and the degree of the right hand side will be inconsistent. 
\section{Feasibility of rate $(1,1)$: The Alternate Proof}
\label{sec:intermediate}
%In this section,  we provide Theorem \ref{mathm}  which  is used to give an intermediate result, in Corollary \ref{cor:alt}, that will be a stepping stone towards deriving the feasibility of rate $(1,1)$ in Theorem \ref{mathm}.
%Now, we state the main theorem.

In this section, we aim to provide an alternate proof to the  the results of \cite{wang-shroff},\cite{CCWang_twounicast}, \cite{rate11-tight}, and \cite{shenvi2009}, which establish the feasibility rate $(1,1)$ for two-unicast networks. In particular, we show that, for any two-unicast-$Z$ network, whenever the generalized network sharing cut set bound   is at least $2$, and the individual source destination pairs have their cut sets of size at least $1$, rate $(1,1)$ is achievable using scalar linear coding. The result is stated in the following theorem.
\begin{theorem}\label{mathm} 
	Consider a $(\mathcal{G},\{s_{1}\},\{t_{1}\},\{s_2\},\{t_2\})$ two-unicast-$Z$ network such that there is a path from  $s_{1}$ to $t_{1}$ and there is a path from $s_{2}$ to $t_{2}$. If $\mathcal{G}$ has a minimum GNS cut set  of size at least two,  then the rate $(1,1)$ is achievable in the network using scalar linear coding. 
\end{theorem}

In order to prove this theorem, we first give an intermediate result.
\subsection{An intermediate result}
In this section, we provide an intermediate result, in Corollary \ref{cor:alt}, which establishes the feasibility of rate $(1,1)$ for a specific class of two-unicast-$Z$ networks before generalizing the feasibility of rate $(1,1)$ for any two-unicast-$Z$ network. First, we introduce the following theorem.

\begin{theorem}\label{bathm} 
	Consider a $(\mathcal{G},\{s_{1}\},\{t_{1}\},\{s_2\},\{t_2\})$ two-unicast-$Z$ network such that there is a path from  $s_{1}$ to  $t_{1}$, there is a path from $s_{2}$ to $t_{2}$, and $\mathcal{G}$ has a minimum GNS cut set $\{e_1, e_2\}$ of size two. If  there is an $s_2 \rightarrow t_{1} \textrm{ via } e_i$ path and an $s_2 \rightarrow t_{1} \backslash e_i$ path for some $i \in \{1,2\}$,  then the rate $(1,1)$ is achievable in the network using scalar linear coding. 
\end{theorem}
\begin{proof}
 %To keep the notation simple and clear, we show the theorem for  $i=1$. Consider a two-unicast-$Z$ network which satisfies the hypothesis of the theorem with $i=1$, i.e., there are  $s_2 \rightarrow t_{1}\textrm{ via }  e_1$ and $s_2 \rightarrow t_{1} \backslash e_1$  paths in the graph. Suppose that the rate $(1,1)$ is not achievable in the two-unicast-$Z$ network. Then, the hypothesis of Lemma \ref{lem0} is satisfied. Therefore,  there exists  a polynomial $P$ such that $P \sum\limits_{p: s_2 \rightarrow t_{1}}\hspace{-0mm}w(p) = \big(\det(\mathbf{M})\big)^{L}$ for some  $L \in \mathbb{Z}_{+}$.   Therefore the hypothesis of Lemma \ref{lem3} holds. Now, because there is an $s_2 \rightarrow t_{1} \textrm{ via } e_1$ path, statement (b) of Lemma \ref{lem3} implies that $ {\operatorname{sumdeg}}_{{\overline{\mathbb{K}}^{(1)}}} \leq L-1$ for any monomial $m$ in $P$. However, because there is an $s_2 \rightarrow t_{1} \backslash e_1$ path, statement (c) of Lemma \ref{lem3} implies that  $ {\operatorname{sumdeg}}_{{\overline{\mathbb{K}}^{(1)}}}=L$, for any monomial $m$ in $P$, which contradicts our previous conclusion based on statement (b) of the lemma. Therefore, $P$ does not contain any nonzero monomials, i.e. $P=0$.  Thus, since $P \sum\limits_{p: s_2 \rightarrow t_{1}}\hspace{-0mm}w(p) = \big(\det(\mathbf{M})\big)^{L}$, we have  $\operatorname{det}(\mathbf{M})=0$, a contradiction to the fact that $\operatorname{det}(\mathbf{M})\neq 0$, Corollary \ref{cor:Mn0}.  
 Consider a two-unicast-$Z$ network which satisfies the hypothesis of the theorem for some  $i\in\{1,2\}$, i.e., there are  $s_2 \rightarrow t_{1}\textrm{ via }  e_i$ and $s_2 \rightarrow t_{1} \backslash e_i$  paths in the network. For contradiction, suppose that the rate $(1,1)$ is not achievable in the two-unicast-$Z$ network using scalar linear coding.  Therefore, by Lemma \ref{lem0}, there exists  a polynomial $P$ such that $P \sum\limits_{p: s_2 \rightarrow t_{1}}\hspace{-0mm}w(p) = \big(\det(\mathbf{M})\big)^{L}$ for some  $L \in \mathbb{Z}_{+}$.  
 
 Now, we investigate the different values of $i$. If $i=1$, this means that there are  $s_2 \rightarrow t_{1}\textrm{ via }  e_1$ paths and $s_2 \rightarrow t_{1} \backslash e_1$  paths in the network. That is, all of $\sum\limits_{p: s_2 \rightarrow e_{1}}\hspace{-0mm}w(p)= \mathbf{b}_1\mathbf{u}_1$, $\sum\limits_{p: e_1 \rightarrow t_{1}}\hspace{-0mm}w(p)
 = \mu_{11}+\mathbf{\lambda}_{12}\mathbf{u}_{2} \mu_{21}$,   $\sum\limits_{p: s_2 \rightarrow e_{2}\backslash e_1}\hspace{-0mm}w(p)=\mathbf{b}_2\mathbf{u}_2$, and $\mu_{21}=\sum\limits_{p: e_2 \rightarrow t_{1}}\hspace{-0mm}w(p)$ are nonzero polynomials.  Now, notice that any monomial in $\mathbf{b}_1\mathbf{u}_1 (\mu_{11}+\mathbf{\lambda}_{12}\mathbf{u}_{2} \mu_{21})$ has sum-degree $1$ in $\overline{\mathbb{K}}^{(1)}$ and any monomial in $ \mathbf{b}_2\mathbf{u}_2 \mu_{21}$ has sum-degree $0$ in  $\overline{\mathbb{K}}^{(1)}$. Recalling that 
$\sum\limits_{p: s_2 \rightarrow t_{1}}\hspace{-0mm}w(p)
 =\mathbf{b}_1\mathbf{u}_1 (\mu_{11}+\mathbf{\lambda}_{12}\mathbf{u}_{2} \mu_{21})+ \mathbf{b}_2\mathbf{u}_2 \mu_{21}$, we conclude that if $i=1$, $\sum\limits_{p: s_2 \rightarrow t_{1}}\hspace{-0mm}w(p)$ is not homogeneous in  $\overline{\mathbb{K}}^{(1)}$. Similarly, if $i=2$, this means that there are  $s_2 \rightarrow t_{1}\textrm{ via }  e_2$ paths and $s_2 \rightarrow t_{1} \backslash e_2$  paths in the network. That is, all of $\sum\limits_{p: s_2 \rightarrow e_{1}}\hspace{-0mm}w(p)= \mathbf{b}_1\mathbf{u}_1$, $\sum\limits_{p: e_1 \rightarrow t_{1}\backslash e_2}\hspace{-0mm}w(p)
 = \mu_{11}$,    $\sum\limits_{p: s_2 \rightarrow e_{2}}\hspace{-0mm}w(p)=\mathbf{b}_1\mathbf{u}_1\mathbf{\lambda}_{12}\mathbf{u}_{2}+ \mathbf{b}_{2}\mathbf{u}_{2}$, and $\sum\limits_{p: e_2 \rightarrow t_{1}}\hspace{-0mm}w(p)=\mu_{21}$ are nonzero polynomials. Now, notice that any monomial in $(\mathbf{b}_1\mathbf{u}_1 \mathbf{\lambda}_{12}\mathbf{u}_{2} +\mathbf{b}_2\mathbf{u}_2)\mu_{21}$ has sum-degree $1$ in $\overline{\mathbb{K}}^{(2)}$ and any monomial in $ \mathbf{b}_1\mathbf{u}_1 \mu_{11}$ has sum-degree $0$ in  $\overline{\mathbb{K}}^{(2)}$. Recalling that 
$\sum\limits_{p: s_2 \rightarrow t_{1}}\hspace{-0mm}w(p)
 =\mathbf{b}_1\mathbf{u}_1 \mu_{11}+(\mathbf{b}_1\mathbf{u}_1\mathbf{\lambda}_{12}\mathbf{u}_{2} + \mathbf{b}_2\mathbf{u}_2) \mu_{21}$, we conclude that if $i=2$, $\sum\limits_{p: s_2 \rightarrow t_{1}}\hspace{-0mm}w(p)$ is not homogeneous in  $\overline{\mathbb{K}}^{(2)}$.
 
This means that, for any $i\in\{1,2\}$, the homogeneous polynomial $\big(\det(\mathbf{M})\big)^{L}$ in  $\overline{\mathbb{K}}^{(i)}$  has a non-homogeneous polynomial (i.e., $\sum\limits_{p: s_2 \rightarrow t_{1}}\hspace{-0mm}w(p)$) in $\overline{\mathbb{K}}^{(i)}$ as a factor, a contradiction to the fact that the factors of any homogeneous polynomial are also homogeneous (Lemma \ref{rmk:homog}). 
 \end{proof}
%\subsection{Feasibility of rate $(1,1)$: The Alternate Proof }
\label{sec:main}
 %Before deriving the alternate proof in Theorem \ref{mathm}.
%, we first eliminate some degenerate cases in Lemma \ref{lemcut}.
%\begin{lemma}\label{lemcut}
%Consider a $(\mathcal{G},\{s_{1}\},\{t_{1}\},\{s_2\},\{t_2\})$ two-unicast-$Z$ network graph such that there is a path from  $s_{1}$ to  $t_{1}$, there is a path from $s_{2}$ to $t_{2}$, and $\mathcal{G}$ has a minimum GNS cut set of size at least two.  If, for some $(i,j) \in \{(1,1),(2,2),(2,1)\}$, we have 
%\begin{itemize}
%\item $s_i=t_j$, or 
%\item $\operatorname{head}(s_i)=\operatorname{tail}(t_j)$,
%\end{itemize}
%then the rate $(1,1)$ is achievable in the network using routing.
%\end{lemma}

%Now, we present the alternate proof of rate $(1,1)$ feasibility in two-unicast-$Z$ networks using  linear coding in the following theorem.

\begin{corollary}\label{cor:alt}
Consider a $(\mathcal{G},\{s_{1}\},\{t_{1}\},\{s_2\},\{t_2\})$ two-unicast-$Z$ network  such that there is a path from  $s_{1}$ to  $t_{1}$, there is a path from $s_{2}$ to $t_{2}$, and $\mathcal{G}$ has a minimum GNS cut set $\{e_1, e_2\}$ of size two. If  the network has two paths that belong to different two of the following classes of $s_2 \rightarrow t_1$ paths,
\begin{enumerate}[label=(\alph*)]
\item the class of $s_2 \rightarrow t_1 \textrm{ via } e_1 \backslash e_2$ paths,
\item  the class of $s_2   \rightarrow t_1 \textrm{ via } e_2 \backslash e_1$ paths, and
\item the class of $s_2  \rightarrow t_1 \textrm{ via } \{e_1 , e_2\}$ paths,
\end{enumerate}   
then, the rate $(1,1)$ is achievable in the network using scalar linear coding. 
\end{corollary}

The proof of the corollary follows directly from Theorem \ref{bathm}. 
\subsection{Proof of Theorem \ref{mathm}}

Inspired by \cite{shannon1956} where the authors define the notion of reduced networks,  we define  critical two-unicast-$Z$ networks.
\begin{definition}[Critical two-unicast-$Z$ network]\label{def:critical}
A $(\mathcal{G},\{s_{1}\},\{t_{1}\},\{s_2\},\{t_2\})$ two-unicast-$Z$ network  with a minimum GNS cut of size  two is  \emph{critical} if removing any edge from the network  reduces the minimum GNS cut size to one. 
\end{definition}
\begin{remark}\label{rmk:critical}
Every edge in a critical two-unicast-$Z$ network belongs to some GNS cut set of size two.
\end{remark}
The remark follows by observing that if any edge in the critical network does not belong to a GNS cut set of size two, then removing this edge does not reduce the size of the minimum GNS cut set of the network to one. That is, the network is not critical, a contradiction. 
\begin{corollary}\label{cor:s2t1}
Consider a critical $(\mathcal{G},\{s_{1}\},\{t_{1}\},\{s_2\},\{t_2\})$ two-unicast-$Z$ network such that there is a path from  $s_{1}$ to  $t_{1}$ and there is a path from $s_{2}$ to $t_{2}$. If, for every GNS cut $\{e_1,e_2\}$ of size two in the network, all the $s_2 \rightarrow t_1$ paths in the network belong to only one class of the following classes of $s_2 \rightarrow t_1$ paths, 
\begin{enumerate}[label=(\alph*)]
\item the class of $s_2 \rightarrow t_1 \textrm{ via } e_1 \backslash e_2$ paths,
\item  the class of $s_2   \rightarrow t_1 \textrm{ via } e_2 \backslash e_1$ paths, and
\item the class of $s_2  \rightarrow t_1 \textrm{ via } \{e_1 , e_2\}$ paths,
\end{enumerate}
 then there is only one $s_2 \rightarrow t_1$ path in the network. 
\end{corollary}
\begin{proof}
For contradiction, assume that there exist more than one $s_2 \rightarrow t_1$ paths in the critical network, pick any two of such paths, and let one of them be named $p_1$ and the other be named $p_2$. Now, pick an edge $e_1$ that belongs to $p_1$ and does not belong to $p_2$ (such an edge exists since $p_1\neq p_2$), and form a GNS cut set of size two that contains $e_1$ (such a GNS cut set of size two exists since the network is critical, Remark \ref{rmk:critical}), let this GNS cut set be $\{e_1,e_2\}$. This means, for the GNS cut set $\{e_1,e_2\}$, there exists an $s_2 \rightarrow t_1$ path in the network that goes through $e_1$ (i.e. $p_1$). Notice that $p_1$  is either an $s_2 \rightarrow t_1$ via $e_1 \backslash e_2$ path or an $s_2 \rightarrow t_1$ via $\{e_1,e_2\}$ path in the network (i.e., $p_1$ belongs to the first or the third class of $s_2 \rightarrow t_1$ paths stated in the corollary). Moreover, $p_2$ is an $s_2 \rightarrow t_1 \backslash e_1$ path. Since the GNS cut set $\{e_1,e_2\}$ cuts every $s_2 \rightarrow t_1$ path and $e_1$ does not belong to $p_2$, $e_2$ cuts $p_2$ (i.e., $e_2$ belongs to $p_2$). Thus, $p_2$ is an $s_2 \rightarrow t_1$ via $e_2 \backslash e_1$ path in the network (i.e., $p_2$ belongs to the second class of $s_2 \rightarrow t_1$ paths stated in the corollary), a contradiction to the hypothesis of the corollary that all the $s_2 \rightarrow t_1$ paths in the network belong to only one class of the $s_2 \rightarrow t_1$ paths. 
\end{proof}
%Before we state Theorem \ref{mathm}, we give the following definition and lemma. 
%\begin{definition}[Inner edges]
% Consider a graph $\mathcal{G}=(\mathcal{V},\mathcal{E})$ where the edges in $\mathcal{E}$ follow some topological ordering. An inner edge in $\mathcal{G}$ is any edge whose topological order lies strictly  between the smallest and the largest topological order in the graph.
%\end{definition}

%\begin{lemma}\label{lemcut1}
%Let $(\mathcal{G},\{s_{1}\},\{t_{1}\},\{s_2\},\{t_2\})$ be a two-unicast-$Z$ network  such that there is a path from  $s_{1}$ to  $t_{1}$, there is a path from $s_{2}$ to $t_{2}$, and $\mathcal{G}$ has a minimum GNS cut set of size at least two.  If there does not exist a GNS cut in the network such that it only consists of inner edges, then the rate $(1,1)$ is achievable in the network using routing.
%\end{lemma}

%The proof of this lemma is in Appendix F.
\begin{lemma}\label{lem:join}
Consider a $(\mathcal{G},\{s_{1}\},\{t_{1}\},\{s_2\},\{t_2\})$ two-unicast-$Z$ network such that there is a path from  $s_{1}$ to $t_{1}$, there is a path from $s_{2}$ to $t_{2}$, and there is only one path from $s_2$ to $t_1$. If an $s_1 \rightarrow t_1$ path joins the $s_2 \rightarrow t_1$ path, it cannot leave it. Similarly,  if an $s_2 \rightarrow t_2$ path leaves the $s_2 \rightarrow t_1$ path, they cannot rejoin.    
\end{lemma}
\begin{proof}
The lemma follows from noticing that if an $s_1 \rightarrow t_1$ path that joined the $s_2 \rightarrow t_1$ path left it, or if an $s_2 \rightarrow t_2$  path that left the $s_2 \rightarrow t_1$ path rejoined it, then the network would contain two different $s_2 \rightarrow t_1$ paths, a contradiction to the fact that the network has only one $s_2 \rightarrow t_1$ path.  
\end{proof}
Now, we can introduce the proof of Theorem \ref{mathm}.
\begin{proof}[Proof of Theorem \ref{mathm}:] In our proof, we assume, without loss of generality, that the two-unicast-$Z$ network is critical with a minimum GNS cut of size two. Indeed, consider any    two-unicast-$Z$ network such that there is a path from  $s_{1}$ to $t_{1}$ and there is a path from $s_{2}$ to $t_{2}$ with minimum GNS cut set  of size at least two, call this network the \emph{original} network. If this original network is not critical, then edges can be removed iteratively till the point such that every edge in the resultant graph, denoted by $\mathcal{G}'$, belongs to some GNS cut of size two, i.e., the resultant graph $\mathcal{G}'$ is critical.

Now, if   rate $(1,1)$ is achievable in the critical network $(\mathcal{G}',\{s_{1}\},\{t_{1}\},\{s_2\},\{t_2\})$ using scalar linear coding, then $(1,1)$ is achievable in the original network $(\mathcal{G},\{s_{1}\},\{t_{1}\},\{s_2\},\{t_2\})$ using scalar linear coding. Moreover, in our proof, we  assume that the critical network has at least one $s_2  \rightarrow t_1$ path, i.e. interference at  $t_1$. Otherwise, the network has two edge disjoint $s_1 \rightarrow t_1$ and $s_2 \rightarrow t_2$ paths, and the rate (1,1) achievability directly follows using routing.  
\begin{figure*}[]
  \subfloat[Case 1:  $e'$ is a single edge GNS cut set.]{%
      \includegraphics[width=0.4\textwidth]{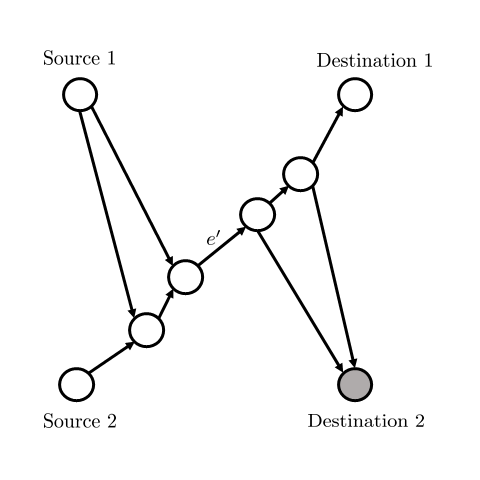}}\hfill
  \subfloat[Case 2:  Rate $(1,1)$ is achievable by routing.]{%
      \includegraphics[width=0.4\textwidth]{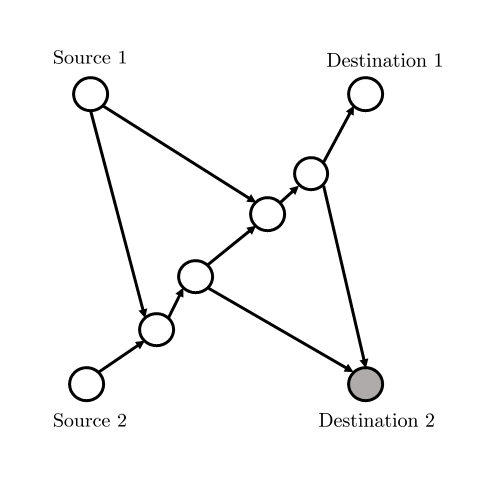}}\hfill
  \caption{Two examples of critical two-unicast-$Z$ networks for the two possible cases: case 1 where all the $s_1\rightarrow t_1$  paths join the only $s_2 \rightarrow t_1$ path in the network  before any $s_2 \rightarrow t_2$ path   leaves the $s_2 \rightarrow t_1$ path, and case 2 where an $s_1 \rightarrow t_1$ path joins the $s_2 \rightarrow t_1$ path after an $s_2 \rightarrow t_2$ path leaves the $s_2\rightarrow t_1$ path.}
\label{fig:one_s2_t1} 
\end{figure*}
If the critical network has two paths that belong to different two of the following classes of $s_2 \rightarrow t_1$ paths:
1) the class of $s_2 \rightarrow t_1 \textrm{ via } e_1 \backslash e_2$ path,
2)  the class of $s_2   \rightarrow t_1 \textrm{ via } e_2 \backslash e_1$ paths, and
3) the class of $s_2  \rightarrow t_1 \textrm{ via } \{e_1 , e_2\}$ paths, where $\{e_1,e_2\}$ is any GNS cut set of size two, 
then, from Corollary \ref{cor:alt}, rate $(1,1)$ is achievable in the network using scalar linear coding. Otherwise, for every GNS cut $\{e_1,e_2\}$ of size two in the network, all the $s_2 \rightarrow t_1$ paths in the network belong to only one class of the following classes of $s_2 \rightarrow t_1$ paths: 1)  the class of $s_2 \rightarrow t_1 \textrm{ via } e_1 \backslash e_2$ paths, 2)  the class of $s_2   \rightarrow t_1 \textrm{ via } e_2 \backslash e_1$ paths, and 3) the class of $s_2  \rightarrow t_1 \textrm{ via } \{e_1 , e_2\}$ paths,
 then, from Corollary \ref{cor:s2t1}, there is only one $s_2 \rightarrow t_1$ path in the network.  Let $p_1$ be the last $s_1 \rightarrow t_1$ path to join this $s_2 \rightarrow t_1$ path. Similarly,  let $p_2$ be the first $s_2 \rightarrow t_2$ path to leave the $s_2 \rightarrow t_1$ path. Now, we have two cases: \textbf{Case 1}:  $p_1$ joins the $s_2 \rightarrow t_1$ path before $p_2$ leaves the $s_2 \rightarrow t_1$ path. In this case, let $e'$ be the first edge in the intersection of $p_1$ and the $s_2 \rightarrow t_1$ path, then, from Lemma \ref{lem:join}, every $s_1 \rightarrow t_1$ path and every $s_2 \rightarrow t_2$ path go through $e'$. Therefore, $e'$ is a single edge GNS cut in the network, a contradiction to the fact that the network has a minimum GNS cut of size two, implying that case 2 must be true which is as follows: \textbf{Case 2}: $p_1$ joins the $s_2 \rightarrow t_1$ path after  $p_2$ leaves the $s_2 \rightarrow t_1$ path,  or $p_1$ joins  the $s_2 \rightarrow t_1$ and $p_2$ leaves  the $s_2 \rightarrow t_1$ at the same node. In this case, $p_1$ and $p_2$ are edge disjoint and the rate (1,1) is achievable by routing.  Examples of critical networks for cases 1 and 2 are shown in Fig. \ref{fig:one_s2_t1}. 
 \end{proof}
\vspace{-0pt}
\section{Insufficiency of edge cut bounds and scalar linear network codes}
\label{sec:GNS}
\begin{figure}[ht]
    \centering
    \includegraphics[scale=0.45]{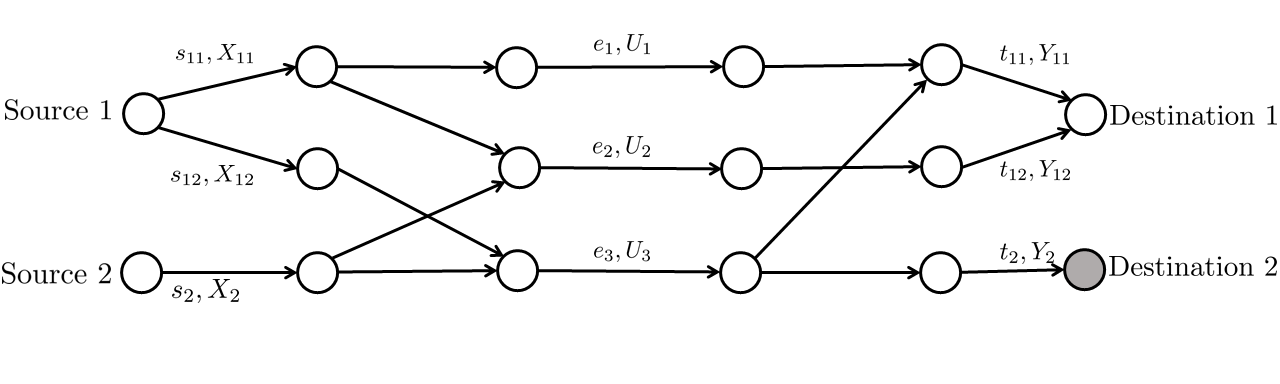}
    \caption{The two-unicast-$Z$ instance $\mathcal{I}$ where GNS bound is $3$ and the maximum achievable sum-rate is $2.5$; the network requires vector linear codes.}
    \label{fig:counterex}
\vspace{-0pt}
\end{figure}
In this section, we show that the generalized network sharing (GNS) bound is not tight for two-unicast-$Z$ networks, and that vector linear codes outperform scalar linear codes in two-unicast-$Z$ networks. We prove these results by constructing a  two-unicast-$Z$ instance where both the GNS bound is not tight and vector linear codes outperform scalar linear codes. 

\subsection{Insufficiency of GNS bound}
The main result of this section regarding the insufficiency of the GNS bound is formulated in  the following theorem.
\begin{theorem}\label{thm:GNSnottight}
There exists a two-unicast-$Z$ instance   where, for any  rate $(R_1,R_2)$ in its rate region, the sum-rate $R_1+R_2$ is strictly less than the cardinality of any GNS cut set. That is, for  the two-unicast-$Z$ network,  the GNS bound is not tight.
\end{theorem}
\vspace{3mm}

In order to prove Theorem \ref{thm:GNSnottight}, we aim to construct a two-unicast-$Z$ instance in which there is a gap between the maximum  sum-rate $R_1+R_2$, over all rates $(R_1,R_2)$ in the rate region of this instance,  and the minimum size of a GNS cut set. In the following, we prove that a candidate instance is the two-unicast-$Z$ instance $\mathcal{I}$ depicted in Fig. \ref{fig:counterex}. First, we establish  an upper bound on the achievable sum-rates in $\mathcal{I}$.

\begin{claim}\label{cl:sumrate}
Let $(R_1,R_2)$ belong to the rate region of the two-unicast-$Z$ instance $\mathcal{I}$ depicted in Fig. \ref{fig:counterex}, $R_1+R_2 \leq 2.5$
\end{claim}

\begin{proof}
Consider a scheme that sends symbols of block length $n$  with probability of error bounded by $\epsilon$. Let $X_{1j}$ denote the symbol sent along edge $s_{1j}, j=1,2$, and $X_2$ denote the symbol sent by $s_2$. Let $Y_{1j}$ denote the symbol received along $t_{1j}, j=1,2$, and $Y_2$ denote the symbol received along $t_2$, and  let $U_{i}$ denote the symbol sent along $e_i, i=1,2,3$. Notice that there is no loss of generality in assuming that $U_1 = X_{11}, U_{2}=Y_{12}, U_3=Y_2$. For the first source, we have
\begin{align*}
	n R_{1} - \epsilon &\leq  I(X_{11},X_{12}; Y_{11},U_2)\\
	                             &\stackrel{(1)}{=} I(X_{11}; Y_{11}) + I(X_{11}; U_{2} | Y_{11})
	                              + I(X_{12};Y_{11}, U_{2}| X_{11}),
\end{align*}
where (1) follows from the application of the chain rule of mutual information.

 Similarly, for the second source, noting that $X_{11},X_{12}$ are available at the second destination as side information, we have
\begin{align*}
	n R_{2} - \epsilon &\stackrel{(2)}{\leq} I(X_{2}; U_{3}|X_{11},X_{12}).
\end{align*}

In addition, we also have
\begin{align*}
	I(X_{11};Y_{11}) + I(X_{11};U_2|Y_{11})&\stackrel{(3)}{=}I(X_{11}; Y_{11},U_2)\stackrel{(4)}{\leq} H(X_{11}) \stackrel{(5)}{\leq} n, 
\end{align*}
where (3) follows from the chain rule of mutual information, (4) follows from the fact that $I(X_{11}; Y_{11},U_2)= H(X_{11})- H(X_{11} |  Y_{11},U_2) $ where $H(X_{11} |  Y_{11},U_2) $ is non-negative, and (5) follows from the edge capacity constraint.
In addition, we can write
\begin{align*}
	I(X_{11};Y_{11}) + I(X_{12};Y_{11},U_2|X_{11}) 
	&= I(X_{11};Y_{11})  + I(X_{12};U_2|X_{11})+ I(X_{12};Y_{11}|X_{11},U_2)\\
	&\stackrel{(6)}{=} I(X_{11};Y_{11})  + I(X_{12};Y_{11}|X_{11},U_2) \\ 
	& \leq  I(X_{11};Y_{11}) + I(X_{12};Y_{11}|X_{11},U_2) \\
&\hspace{10pt}+I(U_2;Y_{11}|X_{11})+I(X_2; Y_{11}| X_{11},U_2,X_{12}) \\
&=I(X_{11},U_2,X_{12},X_2;Y_{11})\\
&\leq H(Y_{11})\\
& \stackrel{(7)}{\leq} n,
\end{align*}
where (6) follows since $X_{12}$ and $U_2$ are independent given $X_{11}.$
 Moreover, we have
\begin{align*}
	I(X_{11};U_2|Y_{11}) +I(X_{12};Y_{11}, U_{2}|X_{11})&=I(X_{11};U_2|Y_{11}) +I(X_{12};Y_{11}|X_{11}) +I(X_{12};U_{2}|X_{11},Y_{11})\\
	& =  I(X_{11},X_{12};U_2|Y_{11})+ I(X_{12};Y_{11}|X_{11})\\
	&\stackrel{(8)}{=}   I(X_{11},X_{12};U_2|Y_{11}) +I(X_{12};U_{3}|X_{11}) \\
	& \leq n+ I(X_{12};U_{3}|X_{11})\\
	&= n + I(X_{12},X_2;U_3|X_{11}) -  I(X_{2};U_{3}|X_{11},X_{12})\\
	& \stackrel{(9)}{\leq} 2n  - I(X_{2};U_{3}|X_{11},X_{12}),
\end{align*}
    where (8) follows from the fact that $H(X_{12}|X_{11},Y_{11})=H(X_{12}|X_{11},U_{3})$, therefore 
    \begin{align}
     I(X_{12};Y_{11}|X_{11})&=  H(X_{12}|X_{11})- H(X_{12}|X_{11},Y_{11})\notag\\
     &=H(X_{12}|X_{11})- H(X_{12}|X_{11},U_{3})\notag\\
     &=I(X_{12};U_{3}|X_{11}).
    \end{align}

Finally, performing $2 \times (1) + (2) + (5) + (7) + (9)$ and letting $n \to \infty$ gives $2R_1+R_2 \leq 4$. In conjunction with the cut set bound on the achievable rate of every source-destination communication session \cite{cutsetbnd}, i.e., $R_1 \leq 1$ and $R_2 \leq 2$, we infer that $R_1+R_2 \leq 2.5$. 
\end{proof}
Now, we prove Theorem \ref{thm:GNSnottight}.
\begin{proof}[Proof of Theorem \ref{thm:GNSnottight}:]
%To prove the theorem, proving that there exists a two-unicast-$Z$ instance where the achievable sum-rate is strictly less than the minimum size GNS cut set suffices.
Notice that the two-unicast-$Z$ instance $\mathcal{I}$ shown in Fig. \ref{fig:counterex} has a minimum GNS cut set of size $3$. However, the sum-rate $R_1+R_2$ such that $(R_1,R_2)$ belongs to the rate region is  upper bounded by $2.5$ (Claim \ref{cl:sumrate}), i.e., the sum-rate is strictly less than the minimum size GNS cut set. This completes the proof.
\end{proof}}
\subsection{{Scalar linear codes} vs vector linear codes }
%In this section, we show that scalar linear codes are insufficient to achieve the capacity in the two-unicast-$Z$ network. More specifically,
In this section, we show that vector linear codes (with vectors of dimension $> 1$) outperform scalar linear codes in the two-unicast-$Z$ network.  
%in terms of the maximum  achievable sum-rate using each technique. 
The main result of this section is formulated in the following theorem.

\begin{theorem}\label{thm:sclincodes}
    There exists a two-unicast-$Z$ instance whose capacity is achievable by vector linear codes and not achievable by any scalar linear code. That is,  
    for the two-unicast-$Z$ network, vector linear codes outperform scalar linear codes and scalar linear codes are insufficient to achieve the capacity.
\end{theorem}
\begin{proof}
%First, note that vector linear codes outperforming scalar linear codes in the two-unicast-$Z$ network implies that scalar linear codes are insufficient to achieve the capacity. Therefore, in order to complete the proof, we only need to prove that, for the two-unicast-$Z$ network,  vector linear codes outperform scalar linear codes.
To prove the theorem, we  construct a two-unicast-$Z$ instance (namely, instance $\mathcal{I}$ in Fig. \ref{fig:counterex}) whose capacity is achievable by vector linear codes but not achievable by any scalar linear code. Recall from Claim \ref{cl:sumrate} that $R_1 + R_2 \leq 2.5$ in $\mathcal{I}$, for any $(R_1,R_2)$ in the rate region of $\mathcal{I}$.
The capacity of $\mathcal{I}$, i.e., rate $(1.5,1)$, can be achieved via vector linear network coding. 
An achievability scheme for the rate $(1.5,1)$ in $\mathcal{I}$ using vector linear codes is shown in Fig. \ref{fig:vec}. Thus, vector linear codes achieves the capacity of instance $\mathcal{I}$. However,  restricting to scalar linear codes, rates higher than $(1,1)$ are not achievable. Hence, for instance $\mathcal{I}$, capacity is achievable by vector linear codes but not achievable by any scalar linear code. That is,  for the two-unicast-$Z$ network, vector linear codes outperform scalar linear codes and scalar linear codes are insufficient to achieve the capacity.
\begin{figure}[ht]
    \centering
    \includegraphics[scale=0.6]{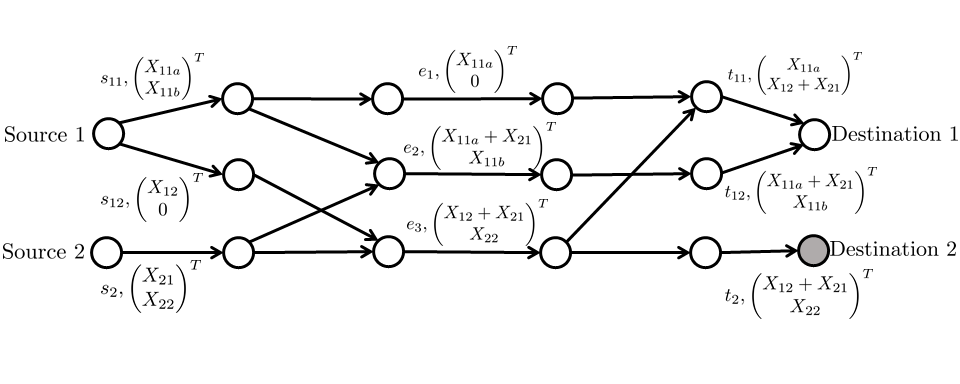}
    \caption{An achievability scheme for rate $(1.5,1)$ using  vector linear codes.}
    \label{fig:vec}
\vspace{-0pt}
\end{figure}
\vspace{-3mm}
\end{proof}

\section{ Insufficiency of Linear Codes}
\label{sec:nonlinear}
\begin{figure}[ht]
    \centering
    \includegraphics[scale=0.5]{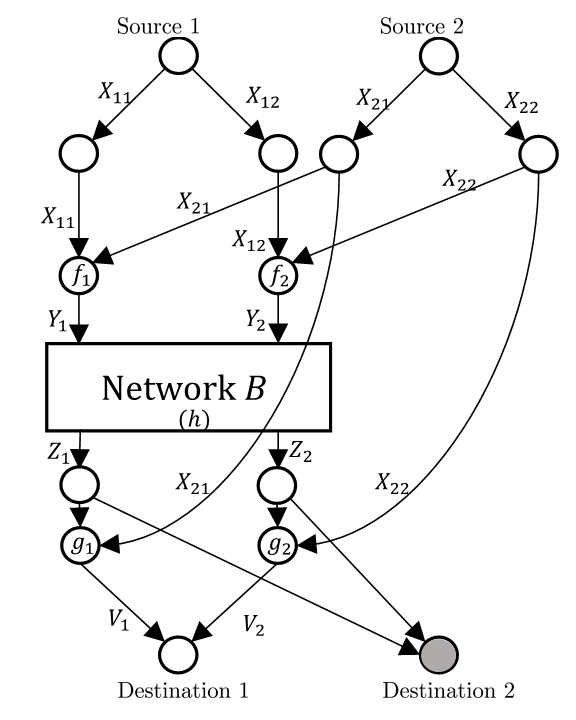}
    \caption{A two-unicast-$Z$ network where rate $(2,2)$ is achievable if and only if rate $(1,1)$ is achievable over two-unicast network $B$.}
    \label{fig:Zhard}
\vspace{-0mm}
\end{figure}
In this section, we show that non-linear codes outperform linear codes in the two-unicast-$Z$ network. In particular, we show that there exists a two-unicast-$Z$ instance where rate $(10,10)$ is not zero-error achievable using linear codes but zero-error achievable using non-linear codes. Our approach is inspired by the method of \cite{Kamath_twounicast}. We consider an arbitrary $m$-unicast network $B$ and construct a two-unicast-$Z$ network, where the zero-error achievability of rate $(m,m)$ in the two-unicast-$Z$ network necessarily requires the zero-error achievability of rate $(1,1, \ldots, 1)$ in the $m$-unicast network $B$.  Since there exists a $10$-unicast instance where linear codes are  insufficient to achieve rate $(1,1, \ldots, 1)$ with zero-error \cite{dougherty2006nonreversibility, dougherty2007matroids, Kamath_twounicast}, our construction implies that linear codes are insufficient to achieve rate $(10,10)$ in two-unicast-$Z$ networks with zero-error. Our construction is shown in Fig. \ref{fig:Zhard}; for simplicity, we describe our method for the special case of $m=2$. The random variables representing the symbol carried by each edge are defined as shown in Fig. \ref{fig:Zhard}.  
%In the following, we exhibit the proof idea used to show that rate $(1,1)$ is achievable in the two-unicast network $B$ if, and only if, rate $(2,2)$ is achievable in the overall two-unicast-$Z$ network. If $(1,1)$ is achievable in the two-unicast network $B$, then simply setting $Y_{i}=X_{1i}+X_{2i}$ and $V_{i} = Z_{i}-X_{2i}$ suffices to achieve $(2,2)$ in the two-unicast-$Z$ network. For the other direction, we show that it is \emph{necessary} for any zero-error achievable scheme with rate $(2,2)$ in the original network, there \emph{must} exist one-to-one mappings from $Y_i$ to $Z_i$, which implies that $(1,1)$ must be achievable in $B$. The intuition for this result is that, to ensure that  $X_{11},X_{12}$ can be decoded from $(V_1, V_2),$, and since destination $1$ has no side information on the second source, both $V_1$ and $V_2$ must be functions of only $X_{11}$ and $X_{12}$. In other words, it is necessary to completely cancel the effect of $X_{22}$ from $Z_1$, and completely cancel the effect of $X_{21}$ from $Z_2$. Since node $g_1$ only has $X_{21}$ as side information, it must be the case that $Z_1$ does not depend on $X_{22}$ and, similarly, $Z_{2}$ does not depend on $X_{21}$. Since $Y_2$ depends on $X_{22}$,  and $Y_1$ depends on $X_{21}$, we require that $Z_{i}$ is simply a function of $Y_{i}$ for $i=1,2$. The rate of the messages dictates that these functions must be one-to-one. 
We formally state our result of the insufficiency of linear codes in the two-unicast-$Z$ networks in the following theorem.
\begin{theorem}
There exists a two-unicast-$Z$ instance in which rate $(10,10)$ is zero-error achievable by non-linear codes but  not zero-error achievable by any linear code. That is,  
    for the two-unicast-$Z$ network, non-linear codes outperform  linear codes and  linear codes are insufficient to achieve the zero-error capacity.
\end{theorem}
\begin{proof}
 For arbitrary $m$, we give a construction in which the zero-error achievability of rate $(m,m)$ in the two-unicast-$Z$ network is equivalent to the zero-error achievability of  rate $(1,1, \cdots, 1)$ in the $m$-unicast network. Since   linear codes are insufficient and non-linear codes are required, in general, to achieve  rate $(1,1, \cdots, 1)$ in $m$-unicast networks for $m > 1$ \cite{ dougherty2007matroids}, our construction implies that linear codes are insufficient  and non-linear codes are required, in general, to achieve $(m,m)$ in two-unicast-$Z$ networks. Our construction is shown in Fig. \ref{fig:Zhard}, for the sake of illustration, we use $m=2$. 

We show that  rate $(1,1)$ is zero-error achievable in the two-unicast network $B$ if, and only if, rate $(2,2)$ is zero-error achievable in the overall two-unicast-$Z$ network; the same idea can be generalized for arbitrary $m$. First, if $(1,1)$ is zero-error achievable over alphabet $\mathcal{A}$ using an $n$ symbol extension in the two-unicast network $B$, then using this scheme for network $B$ in conjunction with setting $Y_{i}=X_{1i}+X_{2i}$ and $V_{i} = Z_{i}-X_{2i}$ is a valid zero-error achievability coding scheme for rate $(2,2)$ in the two-unicast-$Z$ network. Note that here $+$ represents an arbitrary group operation over $\mathcal{A}^{n}$ and $-$ represents its inverse.

For the other direction, let $(2,2)$ be zero-error achievable in the two-unicast-$Z$ network. This implies there exists  a finite alphabet $\mathcal{A}$, a positive integer $n$, and bijective function $f':\mathcal{A}^{2n} \rightarrow \mathcal{A}^{2n}$ between $(X_{11}, X_{12})$ and $(V_1, V_2)$ and, for every $X_{11},X_{12} \in \mathcal{A}^{n},$ bijective functions $g'_{X_{11},X_{12}}:\mathcal{A}^{2n} \rightarrow \mathcal{A}^{2n}$ between $(Z_1,Z_2)$ and $(X_{21},X_{22})$.  

Let $h:\mathcal{A}^{2n}\rightarrow \mathcal{A}^{2n}$ be the relation between the inputs and the outputs of the two-unicast network $B$. That is,   $(Z_1,Z_2)=h(Y_1,Y_2)=(h_1(Y_1,Y_2),h_2(Y_1,Y_2))$ where $h_1,h_2$ are projections of the output of of $h$ on the first and second coordinates respectively.  For a function $f$ on two variables $X,Y$, we use the notation $f|_{X=x}$ to be a function of $Y$ evaluated as $f(X=x, Y)$. Before we prove the result, we make some observations which are consequences of the achievability of $(2,2)$.

{\bf{(1)}} Note that $V_1 = g_1(h_1(f_1(X_{11},X_{21})),X_{21})$ and  $V_2 = g_2(h_2(f_2(X_{12},X_{22})),X_{22})$. Because there exists a bijection from $(V_1, V_2)$ to $(X_{11},X_{12}),$ irrespective of the values of $X_{21},X_{22}$ and because the alphabet of $V_1, V_2$ are each $\mathcal{A}^{n}$, it implies that  $(g_1, g_2)|_{(X_{21},X_{22})=(x_{21}, x_{22})}$ is a surjection on $X_{11},X_{12}$ for all $x_{21},x_{22} \in \mathcal{A}^{n}$. Since the domain of $(g_1, g_2)|_{(X_{21},X_{22})=(x_{21}, x_{22})}$ is $\mathcal{A}^{2n}$, we infer that $(g_1, g_2)|_{(X_{21},X_{22})=(x_{21}, x_{22})}$ is, in fact, a bijection.

{\bf{(2)}}  Using a similar argument as {\bf(1)}, we conclude that  $$\left(g_1(h_1(f_1(X_{11},X_{21})),X_{21})|_{(X_{21}, X_{22})=(x_{21},x_{22})}, g_2(h_2(f_2(X_{12},X_{22})),X_{22})|_{(X_{21}, X_{22})=(x_{21},x_{22})}\right)$$ is a bijection on $X_{11},X_{12}$ for all $x_{21}, x_{22}.$ This implies that, given $x_{11},x_{12}, x_{11}', x_{12}',x_{21},x_{22},$ where $(x_{11},x_{12}) \neq (x_{11}',x_{12}'),$  we have 
$$\left(g_1(h_1(f_1(x_{11},x_{21})),x_{21}), g_2(h_2(f_2(x_{12},x_{22})),x_{22})\right) \neq \left(g_1(h_1(f_1(x_{11}',x_{21})),x_{21}), g_2(h_2(f_2(x_{12}',x_{22})),x_{22})\right).$$

This implies that 
$$\left(f_1(x_{11},x_{21}),f_2(x_{12},x_{22})\right) \neq \left(f_1(x_{11}',x_{21}'), f_2(x_{12}',x_{22}')\right).$$

Thus, we conclude that $(f_1, f_2)|_{(X_{21},X_{22})=(x_{21}, x_{22})}$ is a bijection on $X_{11},X_{12}$ for all $x_{21},x_{22} \in \mathcal{A}^{n}$. 

{\bf{(3)}} Note that $Z_1 = h_1(f_1(X_{11},X_{21}), f_2(X_{12},X_{22}))$ and $Z_2=h_2(f_1(X_{11},X_{21}), f_2(X_{12},X_{22}))$. Because, for every $x_{11},x_{12}$ there exists a bijection from $(Z_1, Z_2)$ to $X_{21},X_{22},$ and because the alphabet of $Z_1, Z_2, Y_1, Y_2$ are each $\mathcal{A}^{n}$, it implies that  $(h_1, h_2)|_{(X_{11},X_{12})=(x_{11}, x_{12})}$ and $(f_1, f_2)|_{(X_{11},X_{12})=(x_{11}, x_{12})}$ are both surjections for all $x_{11},x_{12} \in \mathcal{A}^{n}$. Since $\mathcal{A}^{2n}$ is domain of functions $(h_1, h_2)|_{(X_{11},X_{12})=(x_{11}, x_{12})}$ and $(f_1, f_2)|_{(X_{11},X_{12})=(x_{11}, x_{12})}$, both these functions are, in fact, bijections.

{\bf{(4)}}
Because of {\bf{(2)}} and {\bf{(3)}}, for every $y_1, y_2, y_2' \in \mathcal{A}^{2n}$ where $y_2' \neq y_2$ , there exist $x_{11}, x_{12}, x_{21}, x_{22}, x_{22}'$ where $x_{22} \neq x_{22}'$ such that $y_1 = f_{1}(x_{11}, x_{21}), y_{2} = f_{2}(x_{12},x_{22}), y_{2}'=f_{2}(x_{12},x_{22}').$

We use properties {\bf{(1)}}-{\bf{(4)}} to show that $(1,1)$ is zero-error achievable in the $m$-unicast network $B$ for $m=2.$  To prove that $(1,1)$ is zero-error achievable in $B$, we need to prove that both destinations of network $B$ are satisfied with zero  error probability. In other words, we prove that  $h_1(Y_1,Y_2)=\tilde{h}_{1}(Y_1)$ for all $Y_1, Y_2$,  where $\tilde{h}_{1}:\mathcal{A}^{n}\rightarrow \mathcal{A}^{n}$ is a  bijective function between $Y_1$ and $Z_1$. Similarly, we show that $h_2(Y_1, Y_2) = \tilde{h}_{2}(Y_2)$ for all $Y_1,Y_2$, where $\tilde{h}_{2}$ is a  bijection between $Y_2$ and $Z_2$. %In other words,  we need to prove that $Z_i=h_i(Y_1,Y_2)=h_i(Y_i)$ and $h_i$ is one-to-one mapping, $i= {1,2}$. 

First, we prove that ${h}_i$ depends only on $Y_i$, $i \in \{1,2\}$, that is, we show that $h_1(Y_1, Y_2) = h_1|_{Y_2=y_2}(Y_1)$ and $h_2(Y_1, Y_2) = h_2|_{Y_1=y_1}(Y_2)$ for all $y_1, y_2 \in \mathcal{A}^{n}.$ We show the result for $h_1$, the result for $h_2$ follows by symmetry. Suppose for the sake of contradiction, there exists $y_1, y_{2},y_{2}'$ such that $h_1|_{Y_2=y_2} (y_1) \neq h_1|_{Y_2=y_2'} (y_1).$ By property {\textbf{(4)}} there exist $x_{11}, x_{12}, x_{21}, x_{22}, x_{22}'$ where $x_{22} \neq x_{22}'$ such that $y_1 = f_{1}(x_{11}, x_{21}), y_{2} = f_{2}(x_{12},x_{22}), y_{2}'=f_{2}(x_{12},x_{22}').$

Since $h_1(y_1, y_2) \neq h_1(y_1, y_2')$, by property {\bf{(1)}}, we have $g_1(h_1(y_1, y_2), x_{21}) \neq g_1(h_1(y_1, y_2'), x_{21})$. 

We have thus found $x_{11},x_{12},x_{21},x_{22}'$ such that $$g_1(h_1( f_1(x_{11},x_{21}), f_{2},(x_{12},x_{22})), x_{21}) \neq g_1(h_1( f_1(x_{11},x_{21}), f_{2},(x_{12},x_{22}')), x_{21})$$

This implies that the end-to-end function from $(V_1, V_2)$ to $(X_{11},X_{12})$ is not a bijection, which is a contradiction.

We have thus shown that $h_1(Y_1, Y_2) = h_1|_{Y_2=y_2}(Y_1)$ for all $y_2 \in \mathcal{A}^{n}$ and, by symmetry, $h_2(Y_1, Y_2) = h_2|_{Y_1=y_1}(Y_2)$ for all $ y_1 \in \mathcal{A}^{n}.$ It remains to show that $h_1|_{Y_2=y_2}(Y_1)$ is bijective. It suffices to show that $h_1|_{Y_2=y_2}(Y_1)$ is surjective, since the domain and co-domain of the function are both $\mathcal{A}^{n}.$

Because of property {\bf{(3)}}, the function $(h_1, h_2)|_{(X_{11},X_{12})=(x_{11}, x_{12})}$ is bijective on $X_{21},X_{22}$. Since we have shown that $h_2$ does not depend on $Y_1,$ it does not depend on $X_{11},X_{21}.$ This means, $h_{2}|_{(X_{11},X_{12})=(x_{11}, x_{12})}$ does not depend on $X_{21},$ i.e, $h_{2}|_{(X_{11},X_{12}, X_{22})=(x_{11}, x_{12}, x_{22})}$ is a constant for all $x_{11},x_{12},x_{22}.$ Therefore, it has to be the case that $h_{1}|_{(X_{11}, X_{12},X_{22})=(x_{11},x_{12},x_{22})}$ is bijective on $X_{21}$ for all $x_{11},x_{12},x_{22}$. The function $h_{1}|_{Y_2=y_2}(Y_1)$ must be surjective, since by simply letting $X_{21}$ take all the values in $\mathcal{A}^{n},$ the function $ h_{1}|_{(X_{11}, X_{12},X_{22})=(x_{11},x_{12},x_{22})} = h_{1}|_{(Y_2, X_{11})=(y_2,x_{11})}$ must be able to evaluate to all values in $\mathcal{A}^{n}.$ 

This completes the proof.

\end{proof}

\section{Conclusion and Open Question}\label{sec:conc}
In this paper, we show that the generalized network sharing bound is not tight for two-unicast-$Z$ networks. In addition, we show that, for the two-unicast-$Z$ network, vector linear codes outperform scalar linear codes and non-linear codes outperform linear codes. Another contribution of this paper is introducing a commutative algebraic approach to deriving linear network coding achievability results. This commutative  algebraic approach is demonstrated by providing an alternate proof to the result of C. Wang et. al., I. Wang et. al. and Shenvi et. al. regarding the achievability of rate $(1,1)$ in the network.  

As this paper establishes a relation between the problem of solvability of networks and an equivalent commutative algebraic problem. %An open question to this work includes exploring further the power of the developed commutative algebraic approach in deriving new feasibility results for different networks. Specifically, let $(\mathcal{G}, \mathcal{S}_1, \cdots, \mathcal{S}_m, \mathcal{T}_1, \cdots, \mathcal{T}_n)$ be an arbitrary network defined over the directed acyclic graph $\mathcal{G}=(\mathcal{V},\mathcal{E})$, with sources $\mathcal{S}_i \subset \mathcal{E}$, $i \in \{1, \cdots, m\}$ and $|\mathcal{S}_i|=k_i$ where $k_i$'s are positive integers, and destinations $\mathcal{T}_i \subset \mathcal{E}$, $i \in \{1, \cdots, n\}$. For $i \in \{1, \cdots, m\}$, let the vector $\mathbf{X}_i$ with entries in $\mathbb{K}$ of dimension $k_i$ be the source vector associated with $\mathcal{S}_i$. Let $\sigma$ be some $\{1, \cdots, n\} \rightarrow \{1, \cdots, m\} $ mapping such that, for $i\in \{1, \cdots, n\}$, destination $\mathcal{T}_i$ is interested in the source vector at source $\mathcal{S}_{\sigma{(i)}}$. We assume that $|\mathcal{T}_i|=|\mathcal{S}_{\sigma{(i)}}|$, $i\in \{1, \cdots, n\}$. Let $\mathbf{G}_{i,j}$ be the $|\mathcal{S}_i| \times |\mathcal{T}_j|$ transfer matrix from source $\mathcal{S}_i$ to destination $\mathcal{T}_j$.  In the described network, the rate tuple $(k_1,\cdots, k_m)$ is not achievable if, and only if, there exist polynomials $P_{i,j,r(i),s(j)}$, $(i,j) \in \{1, \cdots, m\} \times \{1, \cdots, n\} - \{(i',j'): i'=\sigma(j')\}$, $r(i) \in \{1\cdots k_i\}, s(j) \in \{1, \cdots, k_j\}$ such that 
%$$ \sum\limits_{\substack{(i,j) \in \{1, \cdots, m\} \times \{1, \cdots, n\}\\ - \{(i',j'): i'=\sigma(j')\} }}\hspace{5mm} \sum\limits_{\substack{r(i) \in \{1, \cdots, k_i\} \\s(j) \in \{1, \cdots, k_j\}}} \mathbf{G}_{i,j}(r(i), s(j)) P_{i,j,r(i),s(j)} = \big(\prod_{(i',j'): i'=\sigma(j')} \operatorname{det}(\mathbf{G}_{i,j})\big)^L,$$
%for some positive integer $L$, where $\mathbf{G}_{i,j}(r(i), s(j))$ is the $(r(i), s(j))$-th entry in $\mathbf{G}_{i,j}$.As this work establishes a relation between the problem of solvability of networks and an equivalent commutative algebraic problem. 
An open question to this work includes exploring further the power of the developed commutative algebraic approach in deriving new feasibility results for different multiple unicast networks, e.g., the two-unicast network.
The two-unicast network has two independent message sources and two destinations, where each destination is interested in one of the two sources. Unlike two-unicast-$Z$ networks, destinations has no apriori side information of any sources in two-unicast networks. Fig. \ref{fig:2unicast} depicts a two-unicast network.
\begin{figure}[ht]
    \centering
    \includegraphics[scale=0.3]{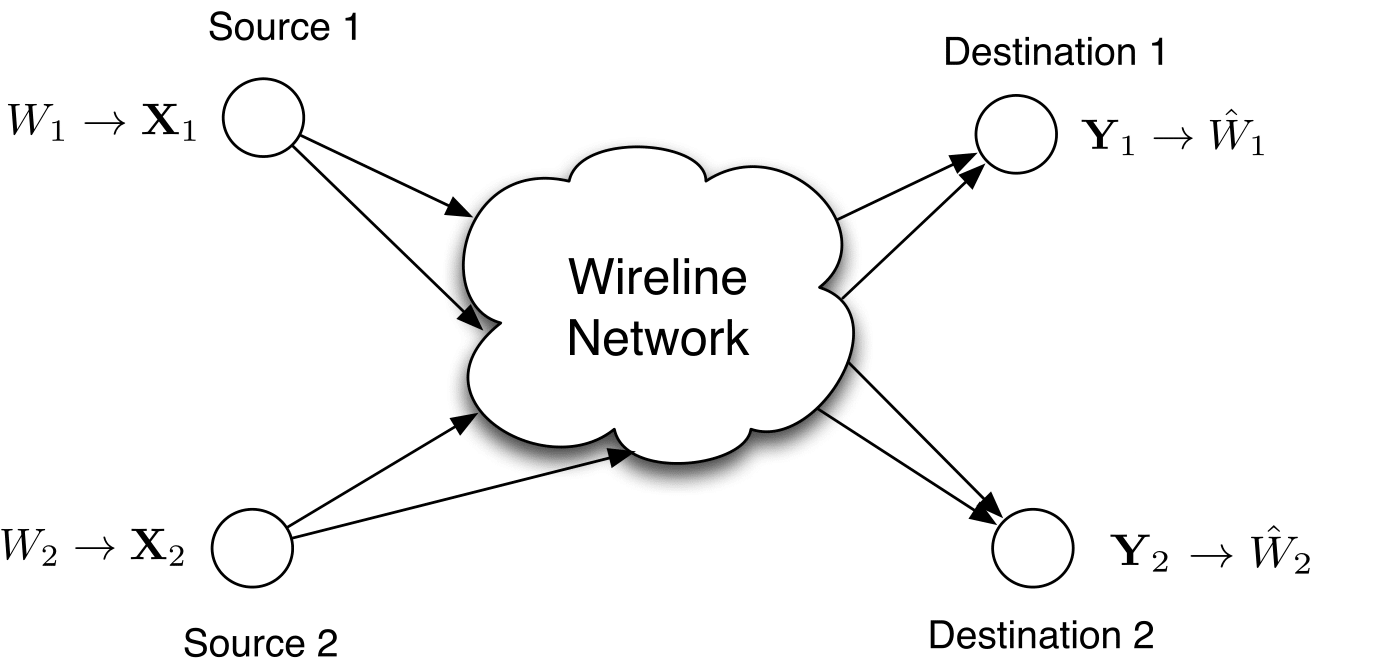}
    \caption{A two-unicast network}
    \label{fig:2unicast}
\vspace{1pt}
\end{figure}
The rate $(1,1)$ in the two-unicast network is achievable if, and only if, there exist polynomials $P_1,P_2$ such that 
\begin{equation} P_1  \mathbf{G}_{1,2} + P_2 \mathbf{G}_{2,1}\label{eq:2un}=\left(\mathbf{G}_{1,1}\mathbf{G}_{2,2}\right)^{L},\end{equation}
for some positive integer $L$.  

The network decomposition lemma proposed in this work applies to two-unicast networks and can be used to \emph{homogenize} the right hand side of (\ref{eq:2un}).  Since the two-unicast network has two interference components, these two components show up in the left hand side of (\ref{eq:2un}). Recalling that the two-unicast-$Z$ network has only one interference component where we were able to conclude the achievability of rate $(1,1)$ based on the non-homogeneity of the interference polynomial, where  monomials had different sum-degrees (Lemma \ref{rmk:homog}), coming up with  a similar conclusion in the presence of multiple interference components is not trivial.
Namely, for the two-unicast network, the challenging problem is to deduce similar degree bounds on the two interference components $\mathbf{G}_{1,2}, \mathbf{G}_{2,1}$ and polynomials $P_1,P_2$  in order to solve the two-unicast network. Such degree bounds may be obtained by investigating graded rings and effective Nullstellensatz.  Solving the two-unicast network using our algebraic perspective will open the way for solving other multiple unicast networks with multiple interference.

\vspace{5pt}
%\section*{Appendix A \\ Proof of Corollary \ref{corol}}

%Here, we prove corollary \ref{corol}. That is, the rate $(1,1)$ is not achievable in a two-unicast-$Z$ network using scalar linear coding if, and only if, for some positive integer $L,$ there exist a polynomial  $P$ such that 
	%\begin{equation}   \mathbf{G}_{2,1} P\label{eq:ach_1}=\left(\mathbf{G}_{1,1}\mathbf{G}_{2,2}\right)^{L}.\end{equation}\label{app:ach_1}
\bibliography{thesis}
\bibliographystyle{abbrv}
\appendices
\section*{Appendix A \\ General Network Decomposition}
In this appendix, we provide a general network decomposition theorem for a  $(\mathcal{G},\mathcal{S}_1,\mathcal{T}_1,\mathcal{S}_2, \mathcal{T}_2)$ two-unicast-$Z$ network with respect to an arbitrary edge subset $\mathcal{U} \subseteq \mathcal{E}$. Before presenting the theorem, we first provide the following definitions.
\begin{definition}[Restricted Transfer matrix]
Consider a DAG $\mathcal{G}=(\mathcal{V},\mathcal{E})$, let $\mathcal{S}'=\{s'_1,s'_2, \cdots, s'_m\}$, $\mathcal{T}'=\{t'_1,t'_2, \cdots, t'_n\}$, $\mathcal{U}=\{u_1,u_2, \cdots, u_k\}$ be any three subsets of $\mathcal{E}$. The transfer matrix $\mathbf{M}^{\mathcal{U}}_{(\mathcal{S}',\mathcal{T}')}$ is defined as the $m \times n$  matrix whose  entry at the index  $(i,j)$ is $\sum\limits_{p: s'_{i} \rightarrow t'_{j} \text{ via } \mathcal{U}} \hspace{-3mm}w(p)$.
\end{definition}
\begin{definition}[Restricted network transfer matrix]
 Consider a $(\mathcal{G},\mathcal{S}_1,\mathcal{T}_1,\mathcal{S}_2, \mathcal{T}_2)$ two-unicast-$Z$ network. Let $\mathcal{U} \subseteq \mathcal{E}$, the restricted network transfer matrix with respect to $\mathcal{U}$ is defined as $\mathbf{M}^{\mathcal{U}}_{(\mathcal{S}_1\cup\mathcal{S}_2,\mathcal{T}_1\cup\mathcal{T}_2)}$.
\end{definition}
\begin{definition}[Destinations-excluded transfer matrix]
Consider a DAG $\mathcal{G}=(\mathcal{V},\mathcal{E})$, let $\mathcal{S}'=\{s'_1,s'_2, \cdots, s'_m\}$ and $\mathcal{T}'=\{t'_1,t'_2, \cdots, t'_n\}$ be any two subsets of $\mathcal{E}$. For any $i,j \in\{1,\cdots, n\}$ such that $i<j$, $\Ord(t'_i)<\Ord(t'_j)$. The destinations-excluded transfer matrix $\mathbf{M}'_{(\mathcal{S}',\mathcal{T}')}$ is defined as the $m \times n$  matrix whose entry at the index  $(i,j)$ is $\sum\limits_{p: s'_{i} \rightarrow t'_{j} \backslash \{t'_k\}_{k<j}} \hspace{-3mm}w(p)$.
\end{definition}
\begin{definition}[Sources-excluded transfer matrix]
Consider a DAG $\mathcal{G}=(\mathcal{V},\mathcal{E})$, let $\mathcal{S}'=\{s'_1,s'_2, \cdots, s'_m\}$ and $\mathcal{T}'=\{t'_1,t'_2, \cdots, t'_n\}$ be any two subsets of $\mathcal{E}$. For any $i,j \in\{1,\cdots, m\}$ such that $i<j$, $\Ord(s'_i)<\Ord(s'_j)$. The sources-excluded transfer matrix $\mathbf{M}''_{(\mathcal{S}',\mathcal{T}')}$ is defined as the $m \times n$  matrix whose entry at the index  $(i,j)$ is $\sum\limits_{p: s'_{i} \rightarrow t'_{j}  \backslash \{s'_k\}_{k>i}} \hspace{-3mm}w(p)$.
\end{definition}
\begin{theorem}[General network decomposition]
\label{lem:netdecomp0}
Consider a $(\mathcal{G},\mathcal{S}_1,\mathcal{T}_1,\mathcal{S}_2, \mathcal{T}_2)$ two-unicast-$Z$ network and let  $\mathcal{U} \subseteq \mathcal{E}$. 
 $\mathbf{M}^{\mathcal{U}}_{(\mathcal{S}_1\cup\mathcal{S}_2,\mathcal{T}_1\cup\mathcal{T}_2)}=\mathbf{M}' _{(\mathcal{S}_1\cup\mathcal{S}_2,\mathcal{U})}\mathbf{\Lambda}^\mathcal{U} \mathbf{M}''_{(\mathcal{U},\mathcal{T}_1\cup\mathcal{T}_2)} $.%, where $\mathbf{M}^\mathcal{U}$ is the network transfer matrix restricted to paths that go through at least one edge in $\mathcal{U}$, $\mathbf{M}_1^\mathcal{U}$ is the transfer matrix from the source edges to $\mathcal{U}$, $\mathbf{M}_2^\mathcal{U}$ is the transfer matrix from  $\mathcal{U}$ to the destination edges,  and   $\mathbf{\Lambda}^\mathcal{U}$ is the coupling matrix of $\mathcal{U}$.
\end{theorem}
\begin{figure}[ht]
    \centering
    \includegraphics[scale=0.55]{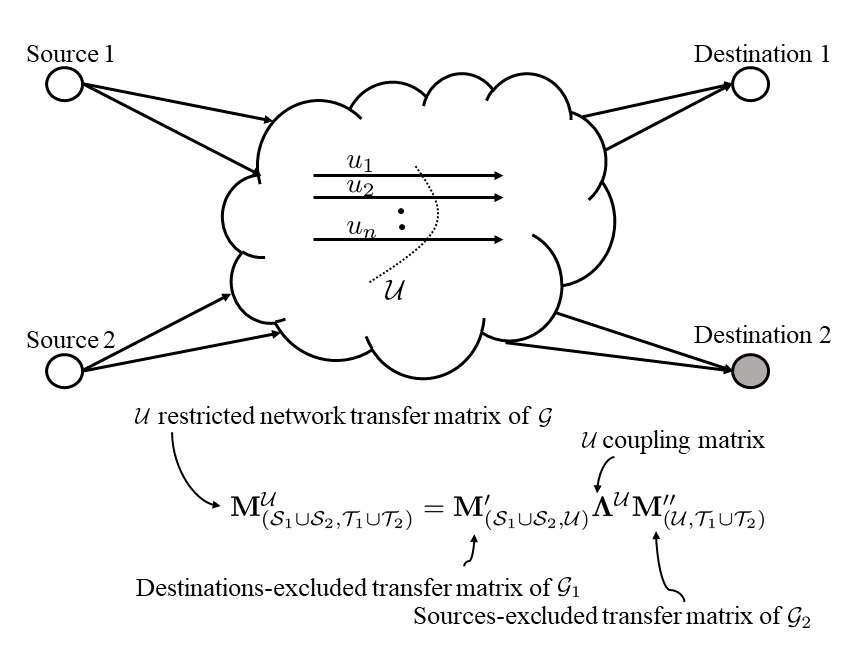}
    \caption{Network decomposition with respect to the subset $\mathcal{U} \subseteq \mathcal{E}$.  }
    \label{fig:decomp}
\vspace{0pt}
\end{figure}

\begin{proof}
A pictorial description of the theorem  is shown in Fig. \ref{fig:decomp}.

 For simplicity, let $\mathcal{U}=\{u_1,u_2\}$ where $\operatorname{Ord}(u_1) < \operatorname{Ord}(u_2)$. Let $s$ and $t$ be two edges $\in \mathcal{E}$, then $\sum\limits_{p: s \rightarrow t \text{ via } \mathcal{U}} \hspace{-3mm} w(p)$ can be decomposed as:
\vspace{0mm}
\begin{align}\label{dec1}
\sum\limits_{p: s \rightarrow t \text{ via }\mathcal{U}} \hspace{-3mm} w(p) &= \sum\limits_{p: s \rightarrow t \hspace{1mm} \text{via} \hspace{1mm} u_1} \hspace{-3mm} w(p) + \sum\limits_{p: s \rightarrow t \hspace{1mm} \text{via} \hspace{1mm} u_2 \backslash u_1} \hspace{-3mm} w(p)\notag\\
&=\sum\limits_{p: s \rightarrow u_1} \hspace{-0mm} w(p) \sum\limits_{p: u_1 \rightarrow t} \hspace{-0mm} w(p)+ \sum\limits_{p: s \rightarrow  u_2 \backslash u_1} \hspace{-3mm} w(p) \sum\limits_{p: u_2 \rightarrow t} \hspace{-0mm} w(p)\notag \\
&= \left(\begin{array}{ccc}  \sum\limits_{p: s \rightarrow u_1} \hspace{-0mm} w(p) &  \sum\limits_{p: s \rightarrow  u_2 \backslash u_1} \hspace{-0mm} w(p) \end{array}\right)\left(\begin{array}{ccc}     \sum\limits_{p: u_1 \rightarrow t} \hspace{-0mm} w(p)  \\    \sum\limits_{p: u_2 \rightarrow t} \hspace{-0mm} w(p)   \end{array}\right)
\end{align}
In addition, $ \sum\limits_{p: u_1 \rightarrow t} \hspace{-0mm} w(p) $ can be expressed as:
\begin{align}
\sum\limits_{p: u_1 \rightarrow t} \hspace{-0mm} w(p) &= \sum\limits_{p: u_1 \rightarrow t \hspace{1mm}  \backslash u_2} \hspace{-3mm} w(p) + \sum\limits_{p: u_1 \rightarrow t \hspace{1mm} \text{via} \hspace{1mm} u_2 } \hspace{-3mm} w(p) \notag\\
&= \sum\limits_{p: u_1 \rightarrow t \hspace{1mm}  \backslash u_2} \hspace{-3mm} w(p) + \sum\limits_{p: u_1 \rightarrow u_2} \hspace{-0mm} w(p)\sum\limits_{p: u_2 \rightarrow t \hspace{1mm}} \hspace{-0mm} w(p)
\end{align} 
Then, we have
\begin{align}
\left(\begin{array}{ccc}     \sum\limits_{p: u_1 \rightarrow t} \hspace{-0mm} w(p)  \\    \sum\limits_{p: u_2 \rightarrow t} \hspace{-0mm} w(p)  \end{array}\right) =& \left(\begin{array}{ccc} 
1 & \sum\limits_{p: u_{1} \rightarrow u_{2} } \hspace{-0mm}w(p)  \\
0 & 1 
\end{array}\right)\left(\begin{array}{ccc} 
\sum\limits_{p: u_{1} \rightarrow t \backslash u_{2}} \hspace{-3mm}w(p)\\
\sum\limits_{p: u_{2} \rightarrow t} \hspace{-0mm}w(p) 
\end{array}\right)
\end{align}
Hence, (\ref{dec1}) can be written as:
\begin{align}
\sum\limits_{p: s \rightarrow t \text{ via } \mathcal{U} } \hspace{-0mm} w(p) &=  \left(\begin{array}{ccc}  \sum\limits_{p: s \rightarrow u_1} \hspace{-0mm} w(p) &  \sum\limits_{p: s \rightarrow  u_2 \backslash u_1} \hspace{-3mm} w(p) \end{array}\right)\left(\begin{array}{ccc} 
1 & \sum\limits_{p: u_{1} \rightarrow u_{2} } \hspace{-0mm}w(p)  \\
0 & 1  \\ 
\end{array}\right)\left(\begin{array}{ccc} 
\sum\limits_{p: u_{1} \rightarrow t \backslash u_{2}} \hspace{-3mm}w(p)\\
\sum\limits_{p: u_{2} \rightarrow t } \hspace{-0mm}w(p)
\end{array}\right)
\end{align}
From the last equation, it is clear that $\mathbf{M}^{\mathcal{U}}_{(\mathcal{S}_1\cup\mathcal{S}_2,\mathcal{T}_1\cup\mathcal{T}_2)}=\mathbf{M}' _{(\mathcal{S}_1\cup\mathcal{S}_2,\mathcal{U})}\mathbf{\Lambda}^\mathcal{U} \mathbf{M}''_{(\mathcal{U},\mathcal{T}_1\cup\mathcal{T}_2)} $.
\end{proof}
\begin{remark}
In this paper,  we consider $\mathcal{U}=\mathcal{C}_{GNS}$, where $\mathcal{C}_{GNS}$ is a GNS cut set in the $(\mathcal{G},\mathcal{S}_1,\mathcal{T}_1,\mathcal{S}_2, \mathcal{T}_2)$ two-unicast-$Z$ network of size two. Specifically, $\mathcal{C}_{GNS}=\{e_1,e_2\}$ where $\Ord(e_1) < \Ord(e_2)$. In addition, for simplicity, we write $\mathbf{M}$, $\mathbf{M}_1$, $\mathbf{\Lambda}$, and $\mathbf{M}_2$ to denote $\mathbf{M}^{\mathcal{U}}_{(\mathcal{S}_1\cup\mathcal{S}_2,\mathcal{T}_1\cup\mathcal{T}_2)}$, $\mathbf{M}' _{(\mathcal{S}_1\cup\mathcal{S}_2,\mathcal{U})}$,  $\mathbf{\Lambda}^\mathcal{U}$, and $\mathbf{M}''_{(\mathcal{U},\mathcal{T}_1\cup\mathcal{T}_2)}$, respectively. 
\end{remark}

\section*{Appendix B \\ Proof of Lemma \ref{lem:netdecomp}}

Here, we prove Lemma \ref{lem:netdecomp}.  We need to prove that for  a $(\mathcal{G},\mathcal{S}_1,\mathcal{T}_1,\mathcal{S}_2, \mathcal{T}_2)$ two-unicast-$Z$ network  with GNS cut set  $\mathcal{C}_{GNS}=\{e_1,e_2\}$,  $\operatorname{Ord}(e_1) < \operatorname{Ord}(e_2)$, we have
\begin{enumerate}[label=(\alph*)]
\item $\mathbf{M}=\mathbf{M_1}\mathbf{\Lambda}\mathbf{M_2}$, where matrices $\mathbf{M},\mathbf{M}_{1},\mathbf{M}_{2},\mathbf{\Lambda}$ are defined by equation (\ref{eq:Ms}).
	\item In graph $\mathcal{G}_{1}$, the network transfer matrix from the source $\{s_{1},s_2\}$ to edges $\{e_1,e_2\}$ is $\mathbf{M}_{1} \mathbf{\Lambda}.$ In graph $\mathcal{G}_{2}$, the network transfer matrix from $\{e_{1},e_{2}\}$ to $\{t_{1},t_{2}\}$ is $\mathbf{\Lambda}\mathbf{M}_{2}$, where $\mathcal{G}_1$ and $\mathcal{G}_2$ are the graphs of the left-side network and right-side network, respectively, with respect to $\mathcal{C}_{GNS}$.
	\item Let $\mathbf{F}_{1} \subset \bar{\mathbf{F}}$ be the set of variables in  $ \sum\limits_{p: s_i \rightarrow e_1} \hspace{-3mm} w(p) )$, $i \in \{1,2\}$, and let $\mathbf{F}_2 \subset \bar{\mathbf{F}}$ be the set of variables in $\mathbf{M}_{2}$, we have $\mathbf{F}_{1} \cap \mathbf{F}_{2} = \phi$. 
	
	\item  Let $\mathbf{F}_{1} \subset \bar{\mathbf{F}}$ be the set of variables in  $ \mathbf{M}_1$, and let $\mathbf{F}_2 \subset \bar{\mathbf{F}}$ be the set of variables in $\sum\limits_{p: e_2 \rightarrow t_i} \hspace{-3mm} w(p)$, $i \in \{1,2\}$, we have $\mathbf{F}_{1} \cap \mathbf{F}_{2} = \phi$. 
	
	\item  Let $\mathbf{F}_{1} \subset \bar{\mathbf{F}}$ be the set of variables in  $ \mathbf{M}_1$, and let $\mathbf{F}_2 \subset \bar{\mathbf{F}}$ be the set of variables in $\mathbf{M}_2$. If there are no $e_1 \rightarrow e_2$ paths in $\mathcal{G}$, then $\mathbf{F}_{1} \cap \mathbf{F}_{2} = \phi$. 
	%\item For any two subsets $\mathbf{F}_{1},\mathbf{F}_{2} \subset \mathbf{F}$ such that $ \sum\limits_{p: s_i \rightarrow e_1} \hspace{-3mm} w(p)  \subseteq \mathbb{K}(\mathbf{F}_{1})$, $i \in \{1,2\}$ and $\mathbf{M}_{2} \subseteq \mathbb{K}(\mathbf{F}_{2})$, we have $\mathbf{F}_{1} \cap \mathbf{F}_{2} = \phi$. 
	
	%\item  For any two subsets $\mathbf{F}_{1},\mathbf{F}_{2} \subset \mathbf{F}$ such that $\mathbf{M}_{1} \subseteq \mathbb{K}(\mathbf{F}_{1})$ and $ \sum\limits_{p: e_2 \rightarrow t_i} \hspace{-3mm} w(p)  \subseteq \mathbb{K}(\mathbf{F}_{2})$, $i \in \{1,2\}$, we have $\mathbf{F}_{1} \cap \mathbf{F}_{2} = \phi$. 
	
	%\item  If there are no $e_1 \rightarrow e_2$ paths, then, for any two subsets $\mathbf{F}_{1},\mathbf{F}_{2} \subset \mathbf{F}$ such that $\mathbf{M}_{1} \subseteq \mathbb{K}(\mathbf{F}_{1})$ and $\mathbf{M}_2  \subseteq \mathbb{K}(\mathbf{F}_{2})$, $i \in \{1,2\}$, we have $\mathbf{F}_{1} \cap \mathbf{F}_{2} = \phi$. 
	\end{enumerate}  
\begin{proof}
(a) The proof of this part follows from Theorem \ref{lem:netdecomp0}, where $\mathcal{U}=\mathcal{C}_{GNS}$.

(b) The proof of this part follows from Theorem \ref{lem:netdecomp0}.

(c) Here, we aim to show that the indeterminate variables that occur in polynomials of $\sum\limits_{p: s_i \rightarrow e_1} \hspace{-3mm} w(p)$, $i\in\{1,2\}$, do not occur in the polynomials of $\mathbf{M}_{2}$ and vice-versa. Let $m_1$ be a monomial that occurs in some polynomial in $\sum\limits_{p: s_i \rightarrow e_1} \hspace{-3mm} w(p)$, $i\in\{1,2\}$. Similarly, let  $m_2$ be a monomial that occurs in some polynomial in $\mathbf{M}_{2}$. It suffices to show that the variables in $m_1$ do not occur in the variables of ${m}_{2}$ and vice-versa. If possible, let $\beta_{e_a,e_b}$ - the local coding coefficient from edge $e_a$ to edge $e_b$ be a variable that occurs in both $m_1$ and $m_2$. We show a contradiction that precludes the existence of $\beta_{e_a,e_b}$.

Notice that $m_1$ is of the form $w(p_1)$ where $p_1$ is some $s_{i}\rightarrow e_1$ path, for some $i \in \{1,2\}$. In addition, from (\ref{eq:Ms}), $m_2$ is of the form $w(p_2)$ where $p_2$ is some $e_{2}\rightarrow t_j$ path, or some $e_{1} \rightarrow t_{j} \backslash e_2$ path, for some $j \in \{1,2\}$. Because $\beta_{e_a, e_b}$ occurs in $w(p_1),$ edges $e_a, e_b$ occur in path $p_1$ which begins at $s_i$ and ends at $e_1$. Therefore, the topological order of $e_a$ is strictly smaller than the topological order of $e_1$. Moreover, because   $\beta_{e_a, e_b}$ occurs in $w(p_2)$ where $p_2$ is some $e_{2}\rightarrow t_j$ path with $\Ord(e_2) > \Ord(e_1)$, or some $e_{1} \rightarrow t_{j} \backslash e_2$ path, for some $j \in \{1,2\}$,  edges $e_a, e_b$ occur in path $p_2$ which begins at $e_2$ and ends at $t_j$ or begins at $e_1$ and ends at $t_j$. Therefore, the topological order of $e_a$ is at least the topological order of $e_1$. Since the topological order of $e_a$ cannot be strictly smaller than the topological order of $e_1$ and at least the topological order of $e_1$  simultaneously, we conclude that such a $\beta_{e_a,e_b}$ variable cannot occur, contradicting our previous assumption.

(d)  Here, we aim to show that the indeterminate variables that occur in polynomials of $\mathbf{M}_1$ do not occur in the polynomials of $\sum\limits_{p: e_2 \rightarrow t_i} \hspace{-3mm} w(p)$, $i\in\{1,2\}$ and vice-versa. Let $m_1$ be a monomial that occurs in some polynomial in $\mathbf{M}_1$. Similarly, let  $m_2$ be a monomial that occurs in some polynomial in $\sum\limits_{p: e_2 \rightarrow t_i} \hspace{-3mm} w(p)$, $i\in\{1,2\}$. It suffices to show that the variables in $m_1$ do not occur in the variables of ${m}_{2}$ and vice-versa. If possible, let $\beta_{e_a,e_b}$ - the local coding coefficient from edge $e_a$ to edge $e_b$ be a variable that occurs in both $m_1$ and $m_2$. We show a contradiction that precludes the existence of $\beta_{e_a,e_b}$. 

Notice that, from (\ref{eq:Ms}), $m_1$ is of the form $w(p_1)$ where $p_1$ is some $s_{i}\rightarrow e_1$ path, or some $s_{i} \rightarrow e_{2} \backslash e_1$ path, for some $i\in \{1,2\}$. In addition, $m_2$ is of the form $w(p_2)$ where $p_2$ is some $e_{2}\rightarrow t_j$ path, for some $j \in \{1,2\}$.   Because   $\beta_{e_a, e_b}$ occurs in $w(p_1)$ where $p_1$ is some $s_{i}\rightarrow e_1$ path with $\Ord(e_1) < \Ord(e_2)$, or some $s_{i} \rightarrow e_{2} \backslash e_1$ path, for some $i \in \{1,2\}$, edges $e_a, e_b$ occur in path $p_1$ which begins at $s_i$ and ends at $e_1$ or begins at $s_i$ and ends at $e_2$, $i \in \{1,2\}$. Therefore, the topological order of $e_b$ is at most the topological order of $e_2$. Moreover, because $\beta_{e_a, e_b}$ occurs in $w(p_2),$ edges $e_a, e_b$ occur in path $p_2$ which begins at $e_2$ and ends at $t_i$, $i \in \{1,2\}$. Therefore, the topological order of $e_b$ is strictly larger than the topological order of $e_2$. Since the topological order of $e_b$ cannot be at most the topological order of $e_2$ and strictly larger than the topological order of $e_2$ simultaneously, we conclude that such a $\beta_{e_a,e_b}$ variable cannot occur, contradicting our previous assumption. 

(e)  Here, we aim to show that the indeterminate variables that occur in polynomials of $\mathbf{M}_1$ do not occur in the polynomials of $\mathbf{M}_{2}$ and vice-versa. Let $m_1$ be a monomial that occurs in some polynomial in $\mathbf{M}_1$. Similarly, let  $m_2$ be a monomial that occurs in some polynomial in $\mathbf{M}_{2}$. It suffices to show that the variables in $m_1$ do not occur in the variables of ${m}_{2}$ and vice-versa. If possible, let $\beta_{e_a,e_b}$ - the local coding coefficient from edge $e_a$ to edge $e_b$ be a variable that occurs in both $m_1$ and $m_2$. We show a contradiction that precludes the existence of $\beta_{e_a,e_b}$. We have the following cases.

Case 1: $p_1$ is some $s_i \rightarrow e_1$, $i \in\{1,2\}$. In this case, from part (c), $m_1$ and $m_2$ do not share any variables.

Case 2:  $p_2$ is some $e_2 \rightarrow t_i$, $i \in \{1,2\}$. In this case, from part (d), $m_1$ and $m_2$ do not share any variables.

Case 3: $p_1$ is an $s_{i} \rightarrow e_2 \backslash e_1$ path and $p_2$ is a $e_{1} \rightarrow t_j\backslash e_2$ path, for some $i,j \in \{1,2\}$. In this case, we show a contradiction that precludes the existence of $\beta_{e_a,e_b}$. Notice that $m_1$ is of the form $w(p_1)$ where $p_1$ is some $s_{i}\rightarrow e_2 \backslash e_1$ path,  for some $i\in \{1,2\}$. In addition, $m_2$ is of the form $w(p_2)$ where $p_2$ is some $e_{1}\rightarrow t_j \backslash e_2$ path, for some $j \in \{1,2\}$.   Because   $\beta_{e_a, e_b}$ occurs in $w(p_2)$ where $p_2$ is some $e_{1}\rightarrow  t_j \backslash e_2$ path, for some $j \in \{1,2\}$, edges $e_a, e_b$ occur in path $p_2$ which begins at $e_1$ and ends at $t_j$. Therefore there exists an $e_1 \rightarrow e_a$ path. In addition, because   $\beta_{e_a, e_b}$ occurs in $w(p_1)$ where $p_1$ is some $s_{i}\rightarrow e_2 \backslash e_1$ path, for some $i \in \{1,2\}$, edges $e_a, e_b$ occur in path $p_1$ which begins at $s_i$ and ends at $e_2$. Therefore there exists an $e_a \rightarrow e_2$ path. The concatenation of the $e_1 \rightarrow e_a$ path and the $e_a \rightarrow e_2$ path gives a $e_1 \rightarrow e_2$ path via $e_a$, a contradiction to the hypothesis that there are no $e_1 \rightarrow e_2$ paths in the graph. Hence, we conclude that such a $\beta_{e_a,e_b}$ variable cannot occur.
\end{proof}

\end{document}